\documentclass{amsart}

\usepackage{amsmath,amssymb}
\usepackage{hyperref}

\newtheorem{thm}{Theorem}
\newtheorem{prop}[thm]{Proposition}
\newtheorem{lem}[thm]{Lemma}
\newtheorem{cor}[thm]{Corollary}

\theoremstyle{definition}
\newtheorem{definition}[thm]{Definition}

\theoremstyle{remark}

\newtheorem{remark}[thm]{Remark}

\numberwithin{equation}{section}

\newcommand{\field}[1]{\mathbb{#1}}

\newcommand{\operator}[1]{\mathsf{#1}}

\newcommand{\Tr}{\operatorname{Tr}}
\newcommand{\D}{\mathcal{D}}

\newcommand{\Char}{\operatorname{Char}}

\newcommand{\W}{\mathcal{W}}
\newcommand{\M}{\mathcal{M}}
\newcommand{\E}{\mathbb{E}}
\newcommand{\Hom}{\operatorname{Hom}}
\newcommand{\End}{\operatorname{End}}
\newcommand{\Mat}{\operatorname{Mat}}
\newcommand{\SD}{\operatorname{SD}}

\def\ra{{\rightarrow}}
\def\tr{{\rm Tr}}

\begin{document}


\title[Asymptotics of Unitary Multimatrix Models]{Asymptotics of Unitary Multimatrix Models: The Schwinger-Dyson Lattice
and Topological Recursion}


\author{Alice Guionnet}
\address{Alice Guionnet, Department of Mathematics, Massachusetts Institute of Technology, 
Cambridge, MA 02139. \hfill
Research partially supported by the Simons foundation and NSF award DMS-1307704 }
\email{guionnet@math.mit.edu  }

\author{Jonathan Novak}
\address{Jonathan Novak, Department of Mathematics, Massachusetts Institute of Technology, 
Cambridge, MA 02139}
\email{jnovak@math.mit.edu}


\begin{abstract}
	We prove the existence of a $1/N$ expansion in unitary
	multimatrix models which are Gibbs perturbations of the Haar measure,
	and express the expansion coefficients recursively in terms of the 
	unique solution of a noncommutative initial value problem.  
	The recursion obtained is closely related to the ``topological recursion''
	which underlies the asymptotics of many random matrix ensembles
	and appears in diverse enumerative geometry problems.
	Our approach consists of two main ingredients: an asymptotic
	study of the Schwinger-Dyson lattice over noncommutative Laurent 
	polynomials, and uniform control on the cumulants of Gibbs measures
	on product unitary groups.
	The required cumulant bounds are obtained by  concentration of measure 
	arguments  and change of variables techniques.
	\end{abstract}


 \maketitle


\tableofcontents


\section{Introduction}  

	\subsection{A noncommutative initial value problem}

	\subsubsection{}
	Given a unital $*$-algebra $B$ defined over $\field{C}$, let

		\begin{equation*}
			L=B\langle u_1^{\pm1},\dots,u_m^{\pm1}\rangle
		\end{equation*}

	\noindent
	denote the algebra of Laurent polynomials in $m$ noncommutative variables
	$u_1,\dots,u_m$, with noncommutative coefficients in $B$.
	That is, 

		\begin{equation*}
			L = B * \field{C}\langle u_1^{\pm1},\dots,u_m^{\pm1}\rangle,
		\end{equation*}

	\noindent
	the free product of $B$ and the group algebra of a free group of
	rank $m$.

	Assuming that the dimension of $B$ is countable, select
	a basis

		\begin{equation*}
			\mathbf{1}=b_0,b_1,b_2,\dots
		\end{equation*}

	\noindent
	in $B$.  The set of reduced words of finite length in the 
	letters 

		\begin{equation*}
			u_1^{\pm1},\dots,u_m^{\pm1},b_1,b_2,\dots
		\end{equation*}

	\noindent
	forms a basis in $L$.  We will reserve the term \emph{monomial}
	for elements of this particular basis.  
	For a norm on $L$, we take the $\ell^1$-norm relative to 
	the monomial basis.

	\subsubsection{}
	In his study of noncommutative analogues of entropy and Fisher
	information \cite{Voiculescu2}, Voiculescu introduced linear maps,
	$\partial_1,\dots,\partial_m$, which act on monomials $p \in L$
	according to the formula

		\begin{equation*}
			\partial_ip = \sum_{p=p_1u_ip_2} p_1u_i \otimes p_2
			- \sum_{p=p_1u_i^{-1}p_2} p_1 \otimes u_i^{-1}p_2.
		\end{equation*}

	\noindent
	In words, $\partial_ip$ is the sum of all simple tensors
	obtained from $p$ by tensoring on the right of a $u_i$, less the sum of all
	simple tensors obtained by tensoring on the left of a $u_i^{-1}$.  

	Viewing the variables $u_1,\dots,u_m$ as coordinates on
	a ``noncommutative $m$-torus'', the maps $\partial_i$ play
	the role of classical partial derivatives on $U(1)^m$, 
	see \cite{Voiculescu2}.  In particular, they annihilate 
	constants,

		\begin{equation*}
			B \subseteq \operatorname{Ker} \partial_i,
		\end{equation*}

	\noindent
	are $B$-bilinear 

		\begin{equation*}
			\partial_i(bpb') = b(\partial_ip)b'
		\end{equation*}

	\noindent
	when $L\otimes L$ is given the natural $B$-bimodule structure,
	and satisfy the product rule

		\begin{equation*}
			\partial_i (pq) = (\partial_ip)(\mathbf{1} \otimes q) +
			(p \otimes \mathbf{1})(\partial_i q)
		\end{equation*}

	\noindent
	when $L \otimes L$ is given the natural algebra structure.
	Linear maps from an algebra into its tensor square with
	these properties are known as \emph{derivation-comultiplications} 
	in free probability \cite{Voiculescu:comult}, 
	and as \emph{double derivations} in noncommutative 
	geometry \cite{CEG}.

	\subsubsection{}
	Consider the noncommutative initial value
	problem

		\begin{equation}
			\label{eqn:InitialValue}
			\left. \begin{aligned}
			\tau \otimes \tau(\partial_i p) &= 0 \\
			\tau|_B &= \sigma
			\end{aligned} \right\},
		\end{equation}

	\noindent
	where $\tau$ is an unknown unital trace on $L$ and 
	$\sigma$ is a given unital trace on $B$.  It is straightforward 
	to establish existence and uniqueness for \eqref{eqn:InitialValue}
	--- indeed, in view of the Liebniz rule, 
	\eqref{eqn:InitialValue} amounts to a recurrence
	reducing the computation of $\tau$ on $L$ to the computation
	of $\sigma$ on $B$.  As a simple example, the reader is invited
	to check that

		\begin{equation*}
			\tau(b_1u_1b_2u_1^{-1}) = \sigma(b_1)\sigma(b_2).
		\end{equation*}

	Let $\tau_\sigma$ denote the unique 
	solution of \eqref{eqn:InitialValue}.  Then, the $*$-subalgebras 
	$$B,\field{C}\langle u_1^{\pm1} \rangle,\dots,\field{C}\langle u_m^{\pm1}\rangle$$
	are $*$-free in the noncommutative probability space 
	$(L,\tau_\sigma)$, see \cite[Proposition 5.17]{Voiculescu2}.  Since 
	free independence has a very concrete combinatorial description \cite{NS}, this
	amounts to a combinatorial rule allowing the efficient computation of $\tau_\sigma(p)$
	for any monomial $p \in L$.

	\subsubsection{}
	It is a fundamental result of Voiculescu that,  if $\sigma$ is the limit of a 
	sequence of matrix traces, then $\tau_\sigma$ is the 
	limit of a sequence of \emph{random} matrix traces \cite{Voiculescu1,Voiculescu2}.  

	Let

		\begin{equation*}
			\rho_N: B \rightarrow \Mat_N(\field{C})
		\end{equation*}

	\noindent
	be a sequence of $*$-representations of $B$
	whose normalized characters approximate $\sigma$,
	in the sense that

		\begin{equation*}
			 \lim_{N \rightarrow \infty} \frac{1}{N}\Tr \rho_N(b) = \sigma(b)
		\end{equation*}

	\noindent
	for each $b \in B$.  Note that, since any homomorphism from a 
	normed $*$-algebra into a $C^*$-algebra is 
	contractive \cite[\S 1.3.7]{Dixmier}, the image of any $b \in B$
	under $\rho_N$ satisfies

		\begin{equation*}
			\|\rho_N(b)\| \leq \|b\|_1,
		\end{equation*}

	\noindent
	where $\|\cdot\|$ is the operator norm on $\Mat_N(\field{C})$.
	For each $N \geq 1$, let $$\mathbf{U}_N=(U_1,\dots,U_m)$$
	be an $m$-tuple of $N \times N$ random unitary matrices 
	drawn independently from Haar measure on
	the unitary group $U(N)$.  For each $p \in L$, denote
	by $\rho_N(p)(\mathbf{U}_N)$ the $N \times N$ random matrix obtained
	by replacing the constants in $p$ according to the representation $\rho_N$, and
	replacing the variables $u_1,\dots,u_m$ with the random matrices
	$U_1,\dots,U_m$.  Then, as shown in \cite{Voiculescu2}, 
	one has

		\begin{equation*}
			\lim_{N \rightarrow \infty} \E  \frac{1}{N}\Tr \rho_N(p)(\mathbf{U}_N) = \tau_\sigma(p)
		\end{equation*}

	\noindent
	for each $p \in L$.  

	A unital trace $\tau$ on a $*$-algebra $A$ is called a
	\emph{character} of $A$ if it is positive, i.e. if $\tau(a^*a) \geq 0$
	for all $a \in A$.
	The approximation of $\tau_\sigma$ by random matrix traces
	clearly implies its positivity.  Thus, from an algebraic point of view, 
	Voiculescu's initial value problem \eqref{eqn:InitialValue} provides 
	a means to induce characters of $L$ from characters of the 
	constant subalgebra $B$.  From a probabilistic perspective,
	one has an algebraic formalism --- asymptotic freeness --- 
	describing the large $N$ asymptotic behaviour of the trace of polynomial
	functions of the $m$-tuple $\mathbf{U}_N$ and the 
	deterministic contractions $\rho_N(b_i)$.

	\subsection{Initial value problem with potential}

	\subsubsection{}
	Collins, Guionnet and Maurel-Segala \cite{CGM} considered
	a noncommutative initial value problem which generalizes 
	\eqref{eqn:InitialValue}, namely 

		\begin{equation}
			\label{eqn:InitialValuePotential}
			\left. \begin{aligned}
			\tau \otimes \tau(\partial_i p) + \tau((\D_iV)p) &= 0 \\
			\tau|_B &= \sigma
			\end{aligned} \right\}.
		\end{equation}

	\noindent
	Here $V \in L$ is a fixed polynomial (the ``potential''),
	and $\D_i$ is the Laurent version of the cyclic derivative of
	Rota, Sagan and Stein \cite{RSS}, i.e. the endomorphism of
	$L$ which acts on monomials according to 

		\begin{equation*}
			\D_ip = \sum_{p=p_1u_ip_2} p_2p_1u_i - \sum_{p=p_1u_i^{-1}p_2} u_i^{-1}p_2p_1.
		\end{equation*}

	\noindent
	In words, $\D_i p$ is the sum of the cyclic shifts of $p$ ending in $u_i$ 
	less the sum of the cyclic shifts of $p$ beginning with $u_i^{-1}$.  
	More functorially, $\D_i = m^{\text{op}} \circ \partial_i$, where
	$m^{\text{op}} \in \Hom(L \otimes L,L)$ is the map which reverses 
	multiplication in $L$.  Note that $\D_i$ is not a derivation of $L$.
	However, it does annihilate $B$, so that \eqref{eqn:InitialValuePotential}
	degenerates to \eqref{eqn:InitialValue} if $V \in B$.

	Although \eqref{eqn:InitialValuePotential} is no longer recursive, the authors of 
	\cite{CGM} established existence and uniqueness of continuous solutions provided that
	the potential $V$ is sufficiently ``close'' to the constant subalgebra
	$B$, in an appropriate sense.  This was done by first proving that \eqref{eqn:InitialValuePotential}
	admits at most one solution via a perturbative argument, and subsequently constructing a solution $\tau_\sigma^V$
	as a limit of traces on \emph{interacting} random unitary matrices whose joint distribution
	is a Gibbs law on $U(N)^m$.

	\subsubsection{}
	As above, let $\rho_N$ be a sequence of matrix
	representations of $B$ whose characters approximate $\sigma$.
	Consider the unit-mass measure Borel measure 
	$\mu_N^V$ on $U(N)^m$ defined by 

		\begin{equation}
			\label{eqn:GibbsEnsemble}
			\mu_N^{V}(\mathrm{d}\mathbf{U}) = \frac{1}{Z_N^{V}}
			e^{N\Tr \rho_N(V)(\mathbf{U})} \mu_N(\mathrm{d}\mathbf{U}),
		\end{equation}

	\noindent
	where $\mu_N$ is Haar measure and $Z_N^V$ is a normalization constant (the ``partition function''). 
	Note that $\mu_N^V$ is invariant under translations of $V$ by elements of $B$.
	In particular, if $V \in B$ is a constant potential, $\mu_N^V$ degenerates to $\mu_N$.
	We refer to the sequence of measures $\mu_N^V$ as the 
	\emph{Gibbs ensemble} generated by $\rho_N(V)$.  Note that 
	$\mu_N^V$ is, in general, a complex measure.  However, it is a genuine probability measure if the function

		\begin{equation*}
			\mathbf{U} \mapsto \Tr \rho_N(V)(\mathbf{U})
		\end{equation*}

	\noindent
	is real-valued $\mu_N$-almost everywhere on $U(N)^m$.  
	If this condition holds, we say that $\rho_N(V)$ generates 
	a \emph{real Gibbs ensemble}.  In particular, $\rho_N(V)$
	generates a real Gibbs ensemble if $V$ is selfadjoint up
	to cyclic symmetry, i.e. if each monomial in $V^*$ is a 
	cyclic shift of a monomial in $V$.

	Let $\mathbf{U}_N^V=(U_1,\dots,U_m)$
	be an $m$-tuple of $N \times N$ random unitary matrices
	whose joint distribution is the real Gibbs law $\mu_N^V$.
	We then have a family of scalar valued random variables
	given by 

		\begin{equation*}
			\Tr \rho_N(p)(\mathbf{U}_N^V), \quad p \in L, \quad N \geq 1.
		\end{equation*}

	\noindent
	The mean and covariance statistics of this family induce two
	sequences of functionals on $L$:

		\begin{align*}
			\W_{1N}^V(p) &= \E\Tr \rho_N(p)(\mathbf{U}_N^V) \\
			\W_{2N}^V(p_1,p_2) &=  \E\Tr \rho_N(p_1)(\mathbf{U}_N^V) \Tr \rho_N(p_2)(\mathbf{U}_N^V)
			- \E\Tr \rho_N(p_1)(\mathbf{U}_N^V) \E\Tr \rho_N(p_2)(\mathbf{U}_N^V).
		\end{align*}

	\noindent
	Using a change of variables argument,
	it was shown in \cite{CGM} that these functionals satisfy the 
	\emph{Schwinger-Dyson equation},

		\begin{equation}
			\label{eqn:SDfirst}
			\W_{1N}^V \otimes \W_{1N}^V(\partial_ip) + N\W_{1N}^V((\D_iV)p) = - \W_{2N}^V(\partial_ip).
		\end{equation}

	\subsubsection{}
	The existence of the functional equation \eqref{eqn:SDfirst}, which holds
	at finite $N$, explains why solutions of \eqref{eqn:InitialValuePotential}
	are limits of random matrix traces.  Indeed, a straightforward compactness 
	argument shows that the sequence of linear functionals $(N^{-1}\W_{1N}^V)_{N=1}^{\infty}$
	admits a limit point. Furthermore, concentration of 
	measure techniques may be used to demonstrate that 

		\begin{equation}
			\label{eqn:ConcentrationFirst}
			\sup_N |\W_{2N}^V(p_1,p_2)| < \infty
		\end{equation}

	\noindent
	for any $p_1,p_2 \in L$, see \cite{AGZ}[Corollary 4.4.31] or Corollary \ref{concmes}.
	It follows that any limit point $\tau$ of $(N^{-1}\W_{1N}^V)_{N=1}^{\infty}$ is a
	solution of \eqref{eqn:InitialValuePotential}.  Given the existence of 
	a unique solution $\tau_\sigma^V$, one thus obtains the 
	pointwise convergence of $N^{-1}\W_{1N}^V$ to $\tau_\sigma^V$
	on $L$.

	\subsection{Main result: higher cumulants and topological recursion}

	\subsubsection{}
	In this article, we go beyond the framework of \cite{CGM}
	and consider the higher cumulants of several interacting
	random unitary matrices distributed according to a Gibbs law.

	Let $\rho_N(V)$ generate a real Gibbs ensemble
	$\mu_N^V$, and let $\mathbf{U}_N^V$ be an $m$-tuple
	of $N \times N$ random unitary matrices whose distribution
	in $U(N)^m$ is $\mu_N^V$. 
	Consider the mixed moment functionals on $L$ defined by

		\begin{equation*}
			\M_{kN}^V(p_1,\dots,p_k) = \E \prod_{j=1}^k \Tr \rho_N(p_j)(\mathbf{U}_N^V),
		\end{equation*}

	\noindent
	and the corresponding mixed cumulant functionals defined
	recursively by

		\begin{equation*}
			\M_{kN}^V(p_1,\dots,p_k) = \sum_{\pi \in \operatorname{Par}(k)} \prod_{R \in \pi}
			\W_{|R|N}^V(p_r: r \in R),
		\end{equation*}

	\noindent
	where $\operatorname{Par}(k)$ is the lattice of partitions of $\{1,\dots,k\}$, 
	the internal product being over the blocks of a given partition $\pi$.  
	Mixed moments and mixed cumulants contain the same probabilistic 
	information, but cumulants are easier to work with.  For example, 
	for $k \geq 2$, the cumulant $\W_{kN}^V$ vanishes whenever one of its
	arguments lies in $B$ --- we will refer to this property as 
	\emph{$B$-connectedness}, or simply \emph{connectedness}.  We use this term
	because the relation between moments and cumulants can equivalently 
	be expressed using the exponential formula of enumerative combinatorics, which is frequently used
	to pass between possibly disconnected and connected combinatorial structures \cite{Stanley}[Chapter 5].

	Our goal in this paper is to show that, when $\|V\|_1$ is sufficiently small, 
	the rescaled cumulants

		\begin{equation}
			\label{eqn:RenormalizedCumulants}
			\tilde\W_{kN}^V = N^{k-2}\W_{kN}^V
		\end{equation}

	\noindent
	admit an $N \rightarrow \infty$ asymptotic expansion on the asymptotic scale
	$N^{-2}$.  Furthermore, we will give a recurrence relation completely 
	determining all the expansion coefficients in terms of the limit
	of $\tilde\W_{1N}^V$.

	\subsubsection{}
	As in \cite{CGM},
	our approach is based on the method of
	Schwinger-Dyson equations in random matrix theory.  
	Going beyond \cite{CGM}, we consider an entire hierarchy
	of noncommutative partial differential equations obtained 
	recursively from \eqref{eqn:SDfirst} which encodes the asymptotics
	of the higher cumulants $\W_{kN}^V$.  To solve this hierarchy, 
	one has to invert a certain partial differential
	operator acting on $B^\perp$, the space of 
	noncommutative Laurent polynomials with no constant term.
	We will prove the following quantitative result.

	\begin{thm}
		\label{thm:main}
		Let $V \in L$ be selfadjoint up to cyclic symmetry,
		and suppose there exists $K\ge 1$

			\begin{equation*}
				\|V\|_1 < \frac{7}{66} \cdot \frac{1}{\deg(V) 2^{(K-1)\deg V} 12^{\deg(V)}}.
			\end{equation*}

		\noindent
		Let $\rho_N:B \rightarrow \Mat_N(\field{C})$ be
		a sequence of matrix representations whose normalized
		characters admit an $N \rightarrow \infty$ asymptotic expansion to 
		$h$ terms:

			\begin{equation}\label{eqn:ConstantAsymptotics}
				N^{-1} \Tr \rho_N(b) = \sum_{g=0}^h \frac{\sigma_g(b)}{N^{2g}} + o\bigg{(} \frac{1}{N^{2h}}\bigg{)}, \quad b \in B.
			\end{equation}
  
		For each $k \in [ 1, K]$ and all $p_1,\dots,p_k \in B^\perp$,
		the renormalized $k$th cumulant $\tilde\W_{kN}^V$ of the real Gibbs
		ensemble generated by $\rho_N(V)$ admits an $N \rightarrow \infty$ 
		asymptotic expansion to $h\le K-1$ terms:

			\begin{equation}
				\label{eqn:Asymptotics}
				\tilde\W_{kN}^V(p_1,\dots,p_k) = \sum_{g=0}^h \frac{\tau_{kg}^V(p_1,\dots,p_k)}{N^{2g}} + o\bigg{(} \frac{1}{N^{2h}}\bigg{)}.
			\end{equation}

		\noindent
		The expansion coefficients $\tau_{kg}^V$ may be described as follows:

		\begin{enumerate}

				\smallskip
				\item
				$\tau_{10}^V$ is the unique solution of the noncommutative initial value
				problem \eqref{eqn:InitialValuePotential} with $\sigma=\sigma_0$; 

				\smallskip
				\item
				For $k=1$ and $g >0$, 

					\begin{equation*}
						\tau^V_{1g}(p)=-\sum_{\ell=1}^{g-1}\tau_{1\ell}\otimes\tau_{1 (g-\ell)}(\overline{\Delta}(\Xi_{\tau_{10}^V}^V)^{-1} p)
						-\tau_{2 (g-1)}(\overline{\Delta}(\Xi_{\tau_{10}^V}^V)^{-1}{p});
					\end{equation*}

				\smallskip
				\item
				For $k >1$ and $g>0$, we have

				\begin{align*}
					\tau_{kg}^V(p_1,\dots,p_k) &= 
					-\sum_{f=1}^g \tau_{k(g-f)}^V(\overline{\operator{T}}_{\tau_{1f}^V}(\Xi_{\tau_{10}^V}^V)^{-1}p_1,\dots,p_k) \\
					&- \sum_I \sum_{f=0}^g 
					\tau_{(|I|+1)f}^V \otimes \tau_{(k-|I|)(g-f)}^V(\overline{\Delta}(\Xi_{\tau_{10}^V}^V)^{-1}p_1 \# p_I \otimes p_{\overline{I}}) \\
					&-\sum_{j=2}^k \tau_{(k-1)g}^V(\overline{\operator{P}}^{p_j}(\Xi_{\tau_{10}^V}^V)^{-1}p_1,\dots,\hat{p}_j,\dots,p_k)  \\
					&-\tau_{(k+1)(g-1)}^V(\overline{\Delta}(\Xi_{\tau_{10}^V}^V)^{-1}p_1,\dots,p_k),
				\end{align*}

				\noindent
				where the second sum on the right is over all proper nonempty subsets $I$ of $\{2,\dots,k\}$.

			\end{enumerate}

			These recurrences are given in terms of certain linear transformations 
			$\overline{\Delta} \in \Hom(B^\perp,L^{\otimes 2})$ and $\overline{\operator{P}}^\bullet,
			\overline{\operator{T}}_\bullet,\Xi_\bullet^\bullet \in \End B^\perp$ which will be described in the next section.
	\end{thm}

	\begin{remark}
	Note that Theorem \ref{thm:main} is stated only for polynomials $p_1,\dots,p_k$ which have no 
	constant term.  It is always possible to reduce to this case.
	Indeed, if $p = q+r$ with $q \in B^\perp$ and $r \in B$,
	we have

		\begin{equation*}
			\tilde\W_{1N}^V(p) = \tilde\W_{1N}^V(q+r) = \tilde\W_{1N}^V(q)+\tilde\W_{1N}^V(r),
		\end{equation*}

	\noindent
	by linearity.  Since the asymptotics of $\tilde\W_{1N}^V(r)$ are given by \eqref{eqn:ConstantAsymptotics},
	one need only determine the asymptotics of $\tilde\W_{1N}^V(q)$.  
	For $k \geq 2,$
	if $p_i=q_i+r_i$ with $q_i \in B^\perp$ and $r_i \in B$, we have

		\begin{equation*}
			\tilde\W_{kN}^V(p_1,\dots,p_k) = \tilde\W_{kN}^V(q_1+r_1,\dots,q_k+r_k) = \tilde\W_{kN}^V(q_1,\dots,q_k),
		\end{equation*}

	\noindent
	by multilinearity and connectedness.

	The above theorem remains true if the expansion \eqref{eqn:ConstantAsymptotics} 
	is unknown, provided $\sigma_0$ is replaced in the inductive relations by $N^{-1} \Tr \rho_N(b), b\in B$. 
	\end{remark}

	\subsubsection{}
	This paper is part of a broad program in random matrix theory, with contributions
	from many authors, which seeks to determine the asymptotics of both microscopic 
	and macroscopic statistics of various classes of random matrices by leveraging information 
	from an appropriate manifestation of the Schwinger-Dyson equations.  
	For an overview of the microscopic side of the story, the reader is referred to \cite{EY}, while the macroscopic side is surveyed
	in \cite{Guionnet}.  In particular, the recursion for the expansion coefficients
	$\tau_{kg}^V$ given in Theorem \ref{thm:main} is closely related to
	the \emph{topological recursion} of mathematical physics \cite{ACKM,BIZ},
	and its modern re-imagining \cite{EO}. 

	Theorem \ref{thm:main} is the unitary analogue of the 
	theorems of Guionnet and Maurel-Segala \cite{GM1,GM2} and Maurel-Segala \cite{M} on
	the asymptotics of the trace of polynomial functions in several interacting random Hermitian 
	matrices whose joint distribution is a perturbation of the $m$-fold product GUE measure. 
	The present paper complements these results by adapting the SD equations 
	technology to the setting of perturbations of the $m$-fold product CUE measure, hence generalizing \cite{CGM} to all order expansions.  
	An additional feature of the present work is the inclusion of the background algebra $B$, whose basis elements
	act as ``external sources'' from the random matrix viewpoint.  

	\subsection{Asymptotic expansion of the free energy}
	Let $V \in L$ be a potential such that $\rho_N(V)$ generates a real Gibbs ensemble $\mu_N^V$,
	and consider the partition function

		\begin{equation*}
			Z_N^V = \int\limits_{U(N)^m} e^{N\Tr \rho_N(V)(\mathbf{U})} \mu_N(\mathrm{d}\mathbf{U})
		\end{equation*}

	\noindent
	of  $\mu_N^V$.  It was proved in \cite{CGM} that 

		\begin{equation*}
			Z_N^V \sim e^{N^2F_0^V},
		\end{equation*}

	\noindent
	with $F_0^V$ a quantity independent of $N$.
	To make this precise, we introduce the \emph{free energy} of
	$\mu_N^V$, which is by definition the quantity

		\begin{equation}
			\label{eqn:FreeEnergy}
			F_N^V := \frac{1}{N^2} \log Z_N^V.
		\end{equation}

	\noindent
	As a corollary of Theorem \ref{thm:main}, we obtain the following refinement of \cite{CGM} to all orders: 

	\begin{cor} \label{cor:main}
		Under the hypotheses of Theorem \ref{thm:main}, the free energy $F_N^V$
		admits an $N \rightarrow \infty$ expansion to $h\le K-1$ terms:

			\begin{equation*}
				\frac{1}{N^2}\log  Z_N^V=\sum_{g=0}^h \frac{F_g^V}{N^{2g}}+o\left(\frac{1}{N^{2h}}\right).
			\end{equation*}

		\noindent
		The coefficients $F_g^V$ in this expansion depend only on $V$ and the functionals
		$\sigma_0,\dots,\sigma_h$.

	\end{cor}

	In fact, it may be shown that each $F_g^V$ is an analytic function of the
	coefficients of $V$ whose Taylor expansion serves as a generating function enumerating certain
	graphs drawn on a compact orientable surface of genus $g$.
	These embedded graphs, or \emph{maps} as they are known, are similar to the maps enumerated by the 
	expansion of free energy of perturbations of the product GUE measure, except that they
	possess additional edge data.  The $g=0$ case of this expansion was developed in \cite{CGM}. 
	As this graphical description is rather involved, we shall not pursue the detailed development
	of its extension to higher genus in the present paper.  
	We want to stress however that the $F_g^V$'s are absolutely summable series  
	whose coefficients are determined by the restriction of the normalized trace to the $*$-subalgebra 
	$\rho_N(B) \subseteq \operatorname{Mat}_N(\field{C})$. 

		\subsection{Central Limit Theorem}

	A corollary of the above large $N$ expansion is the following central limit theorem:

	\begin{cor} \label{cor:CLT}
	Under the hypotheses of Theorem \ref{thm:main} with $K\ge 1$,  for any selfadjoint polynomial $p$ in $L$, for any $\lambda\in\mathbb R$,

		\begin{equation*}
		\lim_{N\ra\infty} \int\limits_{U(N)^m} e^{\lambda( \Tr(\rho_N(p)({\bf U}) - N\tau^V_{10}(p))} \mu_N^V(\mathrm{d}\mathbf{U})
		= e^{\frac{\lambda^2}{2}\gamma^V(p)}$$
		where $$\gamma^V(p)=  - \tau_{10}^V(\overline{\operator{P}}^{p} (\Xi_{\tau^V_{10}}^V)^{-1}p).
		\end{equation*}
	\end{cor}

	Corollary \ref{cor:CLT} should be compared with the analogous central limit 
	theorem for traces of polynomial functions in several random Hermitian matrices
	whose joint law is a deformation of the product GUE measure, see 
	\cite[Theorem 4.7]{GM2}.

	\subsection{Topological combinatorics}
	In the present article, we content ourselves with the derivation of Theorem \ref{thm:main} and 
	postpone the study of a general combinatorial/topological interpretation of the functionals
	$\tau_{kg}^V$ and the affiliated coefficients $F_g^V$to a future work.  
	We do mention, however, the relation of Theorem 
	\ref{thm:main} to the study of one particularly interesting unitary matrix model, namely
	the Harish-Chandra-Itzykson-Zuber model \cite{GZ,IZ,Matytsin,ZZ}.  

	Let $B=\field{C}\langle x,y\rangle$ be the
	algebra of polynomials in two selfadjoint
	noncommutative variables $x,y$, and set $V_t = xuyu^{-1}$, with $t \in \field{R}$ a real parameter.  Let $\rho_N$ satisfying \eqref{eqn:ConstantAsymptotics}.
	The partition function of the corresponding real
	Gibbs ensemble $\mu_N^{V_t}$ is the HCIZ integral 

		\begin{equation*}
			Z_N^{V_t} = \int\limits_{U(N)} e^{tN\Tr(\rho_N(x)U\rho_N(y)U^{-1})} \mathrm{d}U.
		\end{equation*}

	\noindent
	Theorem \ref{thm:main}, when combined with the results of \cite{CGM} and \cite{GGN}, 
	establishes the following topological expansion of the HCIZ free energy

		\begin{equation*}
			F_N^{V_t} = \frac{1}{N^2} \log Z_N^{V_t}.
		\end{equation*}

	\begin{thm}
		\label{thm:HCIZ}
		For each $t \in (-\frac{7}{ 2^{2(K-1)}19008},\frac{7}{2^{2(K-1)}19008})$, the HCIZ free energy
		admits an $N \rightarrow \infty$ asymptotic expansion to $h\le K-1$ terms:

			\begin{equation*}
				F_N^{V_t} = \sum_{g=0}^h \frac{F_g(t)}{N^{2g}} + o\bigg{(} \frac{1}{N^{2h}}\bigg{)},
			\end{equation*}

		\noindent
		The coefficients $F_g(t)$ in this expansion are analytic in a neighbourhood
		of $t=0$, with Maclaurin series given by 

			\begin{equation*}
				F_g(t) = \sum_{d=1}^{\infty} \frac{t^d}{d!} \sum_{\alpha,\beta \vdash d} (-1)^{\ell(\alpha)+\ell(\beta)}
				\sigma_g(x^\alpha)\sigma_g(y^\beta) \vec{H}_g(\alpha,\beta),
			\end{equation*}

		\noindent
		where the internal sum is over all pairs of partitions $\alpha,\beta \vdash d$,

			\begin{equation*}
				\sigma_g(x^\alpha) = \prod_{i=1}^{\ell(\alpha)} \sigma_g(x^{\alpha_i}), \quad
				\sigma_g(y^\beta) = \prod_{i=1}^{\ell(\beta)} \sigma_g(y^{\beta_i}),
			\end{equation*}

		\noindent
		and the $\vec{H}_g(\alpha,\beta)$'s are the monotone double Hurwitz numbers.
	\end{thm}

	The monotone double Hurwitz number $\vec{H}_g(\alpha,\beta)$ with
	$\alpha,\beta \vdash d$ counts a combinatorially restricted subclass 
	of the set of degree $d$ branched covers of the Riemann sphere by 
	a compact, connected Riemann surface of genus $g$ such that the
	covering map has profile $\alpha$ over $\infty$, $\beta$ over $0$,
	and simple branching over the $r$th roots of unity, where 
	$r=2g-2+\ell(\alpha)+\ell(\beta)$ by the Riemann-Hurwitz formula.  
	For more on monotone double Hurwitz numbers, see \cite{GGN}.
	Theorem \ref{thm:HCIZ} is the perturbative version of an asymptotic
	expansion of the HCIZ free energy conjectured to hold by Matytsin
	in \cite{Matytsin}.

	\subsection{Organization}
	The paper is organized as follows.  

	In Section \ref{sec:Prelims}, we cover necessary preliminaries.
	Most importantly, we introduce a deformation of the $\ell^1$-norm
	on $L$ which will play a crucial role in our analysis.

	Section \ref{sec:InitialValue} treats the noncommutative initial value problem
	\eqref{eqn:InitialValuePotential}.  We prove uniqueness of continuous solutions
	in a perturbative regime via an argument which is more conceptual than
	that employed in \cite{CGM}.  In particular, we
	introduce a pair of partial differential operators 
	acting on $B^\perp$ and deduce uniqueness 
	from the invertibility of these operators.

	Section \ref{sec:SDlattice} introduces the Schwinger-Dyson lattice
	over $L$.  In particular, we give all equations of this hierarchy in 
	explicit form.  We then present a secondary form of the SD lattice
	equations which describes them completely in terms of the first row of
	the lattice and the fundamental operators introduced in Section \ref{sec:InitialValue}.
	This is somewhat similar in spirit to the description of
	classical integrable systems by means of Lax pairs.
	Finally, we introduce the notion of \emph{uniformly bounded solutions}
	of the SD equations.  Uniformly bounded solutions lead to a renormalized
	form of the SD lattice which is well-poised for asymptotic analysis.

	Section \ref{sec:Asymptotics} carries out the asymptotic analysis
	of an abstractly given uniformly bounded solution of the SD lattice.  
	Our treatment is perturbative: we work exclusively in the regime where
	the potential $V$ is ``close'' to the constant subalgebra $B$.  In this regime,
	the fundamental operators which describe the SD lattice are automorphisms
	of the completion of $B^\perp$ in an appropriate norm.  
	We obtain an abstract version of Theorem \ref{thm:main}, listed as
	Theorem \ref{thm:asymptotics} below, which shows how the recursion 
	relations of Theorem \ref{thm:main} arise intrinsically from the structure 
	of the SD lattice, without any reference to random matrices.

	Section \ref{sec:MatrixModels} makes the connection with matrix models.
	Almost by definition, the cumulants of the Gibbs ensemble generated by 
	$\rho_N(V)$ form
	a solution of the SD lattice equations --- however, this solution may not be
	uniformly bounded, so that Theorem \ref{thm:asymptotics} is not a priori
	applicable.  When $\rho_N(V)$ generates
	a \emph{real} Gibbs ensemble, probabilistic tools such as concentration 
	of measure can be brought in to verify uniform boundedness.  These 
	probabilistic arguments are carried out in Section \ref{sec:MatrixModels}.
	The upshot of this analysis is that Theorem \ref{thm:main}
	ultimately emerges as a corollary of its more abstract version, 
	Theorem \ref{thm:asymptotics}.

	In Section \ref{sec:consequences}, we derive a central limit theorem for the trace
	of polynomial functions of the $m$-tuple $\mathbf{U}_N^V=(U_1,\dots,U_m)$
	of $N \times N$ random unitary matrices whose joint distribution in 
	$U(N)^m$ is the Gibbs measure $\mu_N^V$.  We then establish the 
	asymptotic expansion of the free energy of $\mu_N^V$.  Finally, we combine
	Theorem \ref{thm:main} with results from \cite{GGN} to obtain a 
	proof of Theorem \ref{thm:HCIZ}.

	\section{Preliminaries}
		\label{sec:Prelims}

	\subsection{Algebras and characters}
	All algebras in this article are normed 
	unital $*$-algebras defined over $\field{C}$.
	Homomorphisms respect $*$-structure. 

	A \emph{character} of an algebra $A$ is a linear functional which is
	normalized, tracial, and nonnegative:

		\begin{equation*}
			\tau(\mathbf{1}) =1, \quad \tau(ab)=\tau(ba), \quad \tau(a^*a) \geq 0.
		\end{equation*}  

	\noindent
	Characters are also known as \emph{tracial states}.  
	Characters play the role of expectation 
	functionals in noncommutative probability theory.
	The collection of all characters of $A$ forms a convex
	set, denoted $\Char A$.

	\subsection{Constants, scalars, and correlators}
	Following the convention of \cite{RSS}, the term ``constant'' refers to elements of
	the algbera $B$, while ``scalar'' is reserved for elements of the one-dimensional 
	subalgebra $\field{C}\mathbf{1} \subseteq B$.

	For $k \geq 2$, a \emph{connected $k$-correlator} is a symmetric $k$-linear functional 
	on $L$ which is tracial in each argument, and vanishes whenever one of its arguments
	is a constant.

	\subsection{Degree filtration}
	Given a monomial $p \in L$, we define $\deg_i^+(p)$ to be the number of occurrences 
	of the variable $u_i$ in $p$.  Similarly, we denote by $\deg_i^-(p)$ the number of occurrences
	of $u_i^{-1}$ in $p$.  We set

		\begin{equation*}
			\deg_i(p) = \deg_i^+(p) + \deg_i^-(p),
		\end{equation*}

	\noindent
	the number of occurrences of $u_i^{\pm 1}$ in $p$, and 

		\begin{equation*}
			\deg(p) = \sum_{i=1}^m \deg_i(p),
		\end{equation*}

	\noindent
	the number of occurrences of $u_1^{\pm1},\dots,u_m^{\pm1}$ in $p$.
	Note that the degree function is not a valuation --- we have 

		\begin{equation*}
			\deg(p_1p_2) \leq \deg(p_1) + \deg(p_2)
		\end{equation*}

	\noindent
	for all monomials $p_1,p_2 \in L$, but this is not in general an equality due to the 
	possibility of cancellations.

	Let $L_d$ denote the vector subspace of $L$ spanned by
	the monomials of degree at most $d$.  Thus $L_0=B$, and

		\begin{equation*}
			L_0 \subseteq L_1 \subseteq \dots L_d \subseteq\dots,\quad L_kL_l \subseteq L_{k+l},\quad \bigcup_{d=0}^{\infty}L_d=L,
		\end{equation*}

	\noindent
	so that we have a filtration of $L$.  We extend the domain of the degree function by 
	declaring $\deg(p)=d$ for any $p \in L_d$.  The degree filtration does not see the difference
	between constants and scalars.

	In Section \ref{sec:MatrixModels}, we will need the notion of \emph{balanced polynomials}.

	\begin{definition}
		\label{def:balanced}
	A monomial $p \in L$ is said to be \emph{balanced} if 

			\begin{equation*}
				\sum_{i=1}^m \deg_i^+(p) = \sum_{i=1}^m \deg_i^-(p).
			\end{equation*}

	\noindent
	A polynomial is balanced if it is the sum of balanced monomials.
	\end{definition}

	\subsection{Inner product}
	We equip $L$ with the inner product in which the monomials form
	an orthonormal basis.  We then have

		\begin{equation*}
			B^\perp = \bigcup_{d=1}^{\infty} L_d,
		\end{equation*}

	\noindent
	the space of Laurent polynomials with no constant term.

	Any $p \in L$ decomposes as

		\begin{equation*}
			p = \sum_{q} \langle q,p \rangle q,
		\end{equation*}

	\noindent
	where the sum is over the monomial basis in $L$. 
	Our convention is that inner products are linear in the second
	argument.

	\subsection{Parametric norm}
	Let $\xi$ be a positive parameter, and for each $p \in L$ set

		\begin{equation*}
			\|p\|_\xi = \sum_q |\langle q,p \rangle| \xi^{\deg(q)},
		\end{equation*}

	\noindent
	where the summation is over the monomial basis in $L$.

	\begin{prop}
		For any $\xi \geq 1,$ $(L,\|\cdot\|_\xi)$ is a normed $*$-algebra.
	\end{prop}

	\begin{proof}
		We leave it to the reader to check that $\|\cdot\|_\xi$ is a
		vector space norm with 
		 respect to which the involution in $L$ is an isometry.
		We prove that $\|\cdot\|_\xi$ is an algebra norm.  This is where the condition $\xi\geq 1$ is required.
		Indeed			\begin{align*}
				\|p_1p_2\|_\xi &= \bigg{\|} \bigg{(} \sum_{q_1}  \langle q_1,p_1\rangle  q_1 \bigg{)} \bigg{(} \sum_{q_2}   \langle q_2,p_2\rangle  q_2 \bigg{)}  \bigg{\|}_\xi \\
				&\leq \sum_{q_1,q_2}  | \langle q_1,p_1\rangle\langle q_2,p_2\rangle| \|q_1q_2\|_\xi, \quad \text{ since  $\|\cdot\|_\xi$ is a vector space norm,} \\
				&= \sum_{q_1,q_2}  | \langle q_1,p_1\rangle\langle q_2,p_2\rangle|
				  \xi^{\deg(q_1q_2 )} \\
				&\leq  \sum_{q_1,q_2}  | \langle q_1,p_1\rangle\langle q_2,p_2\rangle|
				  \xi^{\deg(q_1)+\deg(q_2 )}, \quad \text{since $\xi\geq1$,} \\
				&= \|p_1\|_\xi \|p_2\|_\xi.
			\end{align*}
	\end{proof}

	The norm $\|\cdot\|_\xi$ is a deformation of the usual $\ell^1$-norm, which 
	is the case $\xi=1$.  Note that, while all monomials are unit vectors
	in the $\ell^1$-norm, this is not the case for $\xi>1$.  In the range $\xi>1$,
	the $\xi$-norm favours monomials of high degree and penalizes 
	monomials of low degree.  

	\begin{prop}
		\label{prop:xiInequality}
		For any $p \in L$ and any $1 \leq \xi_1 \leq \xi_2$, we have

			\begin{equation*}
				\|p\|_{\xi_1} \leq \|p\|_{\xi_2} \leq \|p\|_{\xi_1} \left( \frac{\xi_2}{\xi_1} \right)^{\deg(p)}.
			\end{equation*}
	\end{prop}

	\begin{proof}
		The first inequality is obvious.  For the second, we argue as follows:

			\begin{align*}
				\|p\|_{\xi_2}  &= \sum_q |\langle q,p \rangle| {\xi_2}^{\deg(q)}= \sum_q |\langle q,p \rangle| \xi_1^{\deg(q)} \left( \frac{\xi_2}{\xi_1} \right)^{\deg(q)}\\
					& \leq \left( \frac{\xi_2}{\xi_1} \right)^{\deg(p)} \sum_q |\langle q,p \rangle| \xi_1^{\deg(q)}  = \left( \frac{\xi_2}{\xi_1} \right)^{\deg(p)} \|p\|_{\xi_1},
			\end{align*}

		\noindent
		where to obtain the inequality we used the fact that, by definition, $\deg(q) \leq \deg(p)$ for any
		monomial $q$ appearing in $p$.
	\end{proof}

	\subsection{Continuous functionals and operators}
	A linear functional $\operator{f} \in \Hom(L,\field{C})$ is $\xi$-continuous if and only if there
	exists a constant $C$ such that

		\begin{equation*}
			|\operator{f}(p)| \leq C\|p\|_\xi
		\end{equation*}

	\noindent
	for all $p \in L$.  We denote by $\Hom_\xi(L,\field{C})$ the set of all 
	$\xi$-continuous linear functionals on $L$; it is a vector subspace of $\Hom(L,\field{C})$.
	The $\xi$-norm of $\operator{f} \in \Hom_\xi(L,\field{C})$, denoted $\|\operator{f}\|_\xi$, 
	is the infimum over all $C$ such that the above Lipschitz inequality holds.  We have

		\begin{equation*}
			\xi_1 \leq \xi_2 \implies \Hom_{\xi_1}(L,\field{C}) \subseteq \Hom_{\xi_2}(L,\field{C}).
		\end{equation*}

	A linear operator $\operator{T} \in \End L$ is $(\xi_1,\xi_2)$-continuous if
	and only if there exists a constant $C$ such that

		\begin{equation*}
			\|\operator{T} p\|_{\xi_2} \leq C\|p\|_{\xi_1}
		\end{equation*}

	\noindent
	for all $p \in L$, and $\|\operator{T}\|_{\xi_1,\xi_2}$ is the infimum over all
	$C$ such that this inequality holds.  The set of $(\xi_1,\xi_2)$-continuous 
	linear operators on $L$ is a unital $*$-subalgebra of $\End(L)$ denoted
	$\End_{\xi_1,\xi_2}(L)$.  For any $\xi_0 \geq 1$,

		\begin{equation*}
			\xi_1 \leq \xi_2 \implies \End_{\xi_1,\xi_0}(L) \subseteq \End_{\xi_2,\xi_0}(L).
		\end{equation*}

	\noindent
	We will shorten $\End_{\xi,\xi}(L)$ to $\End_{\xi}(L)$, and refer to elements
	of this algebra as $\xi$-continuous rather than $(\xi,\xi)$-continuous.

	\subsection{Tensor powers}
	We will need the tensor powers of $L$, 

		\begin{equation*}
			L^{\otimes k} = \underbrace{L \otimes \dots \otimes L}_{k \text{ times}}.
		\end{equation*} 

	\noindent
	We equip $L^{\otimes k}$ with the natural algebra structure in which 
	simple tensors are multiplied according to the rule

		\begin{equation*}
			(p_1 \otimes \dots \otimes p_k)(q_1 \otimes \dots \otimes q_k) = p_1q_1 \otimes \dots \otimes p_kq_k.
		\end{equation*}

	\noindent
	Simple tensors all of whose factors are monomials form a basis of this algebra which, by abuse of language,
	we will refer to as the monomial basis of $L^{\otimes k}$.  
	We equip $L^{\otimes k}$ with the inner product in which the monomial basis is orthonormal.

	All of the above constructions for $L$ go through for $L^{\otimes k}$.
	We have the degree function defined by

		\begin{equation*}
			\deg(p_1 \otimes \dots \otimes p_k) = \deg(p_1) + \dots + \deg(p_k),
		\end{equation*}

	\noindent 
	and the corresponding degree filtration in $L^{\otimes k}$.  We also have the 
	corresponding $\xi$-norm on $L^{\otimes k}$, which is defined by 

		\begin{equation*}
			\|T\|_\xi = \sum_{q_1 \otimes \dots \otimes q_k} |\langle q_1 \otimes \dots \otimes q_k,T\rangle| 
			\xi^{\deg(q_1 \otimes \dots \otimes q_k)}
		\end{equation*}

	\noindent
	for all $T \in L^{\otimes k}$, the summation being over the monomial basis of
	$L^{\otimes k}$.  We have the identity

		\begin{equation*}
			\|p_1 \otimes \dots \otimes p_k\|_\xi = \|p_1\|_\xi \dots \|p_k\|_\xi
		\end{equation*}

	\noindent
	for simple tensors in $L^{\otimes k}$.

	By convention, $L^{\otimes0}$ is the line $\field{C}\mathbf{1}$
	in $L$ spanned by the unit element.  
	Note that there is a unique algebra isomorphism $\field{C}\mathbf{1} \rightarrow \field{C}$
	given by $\mathbf{1} \mapsto 1$, and under this identification the $\xi$-norm identifies with the usual
	norm on $\field{C}$ for any value of $\xi$.

	We will consider linear transformations

		\begin{equation*}
			\operator{T}:(L^{\otimes k_1},\|\cdot\|_{\xi_1}) \rightarrow (L^{\otimes k_2},\|\cdot\|_{\xi_2})
		\end{equation*}

	\noindent
	mapping between the various tensor powers of $L$.
	A linear transformation $\operator{T} \in \Hom(L^{\otimes k_1},L^{\otimes k_2})$
	is $(\xi_1,\xi_2)$-continuous if and only if there exists a constant $C$ such that

		\begin{equation*}
			\|T p_1 \otimes \dots \otimes p_{k_1} \|_{\xi_2} \leq C \|p_1 \otimes \dots \otimes p_{k_1}\|_{\xi_1}
		\end{equation*}

	\noindent
	for all monomials $p_1 \otimes \dots \otimes p_{k_1} \in L^{\otimes k_1}$.
	The operator norm of $\operator{T}$, denoted $\|\operator{T}\|_{\xi_1,\xi_2}$,
	can be calculated by infimizing $C$ over the monomial basis.

	Allowing different instances of the $\xi$-norm on the source and target of our
	linear maps is useful for the following reason.  Certain linear transformations which
	we will need to deal with are not $(\xi,\xi)$-continuous for any $\xi \geq 1$,
	but are $(\xi_1,\xi_2)$-continuous, and even contractive, 
	if the ratio $\xi_1/\xi_2$ is large enough.

	\subsection{Completion}
	We denote by $\mathcal{L}_\xi$ the completion of $L$ with respect to the $\xi$-norm.
	Viewing $L$ as the algebra of polynomial functions $p(u_1,\dots,u_m)$ on a noncommutative $m$-torus,
	$\mathcal{L}_\xi$ may be viewed as the algebra of functions $f(u_1,\dots,u_m)$ whose 
	Fourier coefficients $\langle q, f\rangle$ decay faster than $\xi^{\deg(q)}$.  In particular,
	$\xi_1 \leq \xi_2$ implies $\mathcal{L}_{\xi_1} \supseteq \mathcal{L}_{\xi_2}$.

\section{The initial value problem revisited}
\label{sec:InitialValue}

	In this section we consider the noncommutative initial value problem
	\eqref{eqn:InitialValuePotential} and prove uniqueness of solutions
	in a perturbative regime.

	\begin{thm}
		\label{thm:uniqueness}
		Let $\sigma \in \Hom_1(B,\field{C})$ be a unital trace. 
		If $V \in L$ satisfies 

			\begin{equation*}
				\|\Pi V\|_1 < \frac{7}{66} \cdot \frac{1}{\deg(V) 12^{\deg(V)}},
			\end{equation*}

		\noindent
		where $\Pi$ is the orthogonal projection of $L$ onto $B^\perp$,
		then there is at most one unital trace $\tau \in \Hom_1(L,\field{C})$
		which satisfies for all $p\in L$, 

			\begin{equation*}
			\left. \begin{aligned}
			\tau \otimes \tau(\partial_i p) + \tau((\D_iV)p) &= 0 \\
			\tau|_B &= \sigma
			\end{aligned} \right\}.
		\end{equation*}
	\end{thm}

	A non-quantitative version of the same result was obtained in \cite[Theorem 3.1]{CGM}.  
	Here we give a new, more conceptual argument based on the inversion of a certain 
	differential operator acting on noncommutative Laurent polynomials.
	The methods developed in this section will be repeatedly applied in
	the remainder of the paper, and their introduction at an early stage
	clarifies the exposition.

	\subsection{The cyclic gradient trick}
	\label{subsec:CyclicGradientTrick}
	Our approach to the initial value problem \eqref{eqn:InitialValuePotential}
	is based on considering its implications for the coordinates of the
	cyclic gradient of a monomial $p$,

		\begin{equation*}
			\D p = (\D_1p,\dots,\D_mp).
		\end{equation*}

	\noindent
	Any solution $\tau$ of \eqref{eqn:InitialValuePotential} must satisfy

		\begin{equation}
			\label{eqn:CyclicGradient}
			\tau \otimes \tau(\partial_i\D_ip) + \tau((\D_iV)(\D_ip)) = 0, \quad 1 \leq i \leq m.
		\end{equation}

	\begin{prop}
		\label{prop:ReducedLaplacian}
		For any unital trace $\tau$ on $L$, we have

			\begin{equation*}
				\tau \otimes \tau(\partial_i\D_ip) = \tau(\operator{D}_ip) + 
				\tau \otimes \tau(\Delta_ip),
			\end{equation*}

		\noindent
		where $\operator{D}_i \in \End L$ acts on monomials according
		to 

			\begin{equation*}
				\operator{D}_ip = \deg_i(p)p,
			\end{equation*}

		\noindent
		and $\Delta_i \in \Hom(L,L \otimes L)$ acts on monomials
		according to 

			\begin{equation}
				\label{eqn:ReducedLaplacian}
				\begin{split}
				\Delta_ip &= \sum_{p=p_1u_ip_2} \bigg{(} \sum_{p_2p_1u_i=q_1u_iq_2u_i} q_1u_i \otimes q_2u_i
				- \sum_{p_2p_1u_i=q_1u_i^{-1}q_2u_i} q_1 \otimes q_2 \bigg{)} \\
				-&\sum_{p=p_1u_i^{-1}p_2} \bigg{(} \sum_{u_i^{-1}p_2p_1=u_i^{-1}q_1u_iq_2} q_1 \otimes q_2
				- \sum_{u_i^{-1}p_2p_1=u_i^{-1}q_1u_i^{-1}q_2} u_i^{-1}q_1 \otimes u_i^{-1}q_2 \bigg{)} 
				\end{split}
			\end{equation}
	\end{prop}

	\begin{proof}
	Let $p \in L$ be a monomial.  We will expand the tensor $\partial_i\D_ip$ into
	simple tensors. We have

	\begin{equation*}
		\D_ip = \sum_{p=p_1u_ip_2} p_2p_1u_i - \sum_{p=p_1u_i^{-1}p_2} u_i^{-1}p_2p_1,
	\end{equation*}

	\noindent
	the sum of the cyclic shifts of $p$ ending in $u_i$ less the sum of the
	cyclic shifts of $p$ beginning with $u_i^{-1}$.  Applying $\partial_i$,
	this becomes

		\begin{align*}
			\partial_i\D_ip &= \sum_{p=p_1u_ip_2} \partial_ip_2p_1u_i - \sum_{p=p_1u_i^{-1}p_2} \partial_iu_i^{-1}p_2p_1 \\
			&= \sum_{p=p_1u_ip_2} \bigg{(} p_2p_1u_i \otimes \mathbf{1} + \sum_{p_2p_1u_i=q_1u_iq_2u_i} q_1u_i \otimes q_2u_i
				- \sum_{p_2p_1u_i=q_1u_i^{-1}q_2u_i} q_1 \otimes u_i^{-1}q_2u_i \bigg{)} \\
			-&\sum_{p=p_1u_i^{-1}p_2} \bigg{(} \sum_{u_i^{-1}p_2p_1=u_i^{-1}q_1u_iq_2} u_i^{-1}q_1u_i \otimes q_2
				- \sum_{u_i^{-1}p_2p_1=u_i^{-1}q_1u_i^{-1}q_2} u_i^{-1}q_1 \otimes u_i^{-1}q_2 - \mathbf{1} \otimes u_i^{-1}p_2p_1\bigg{)} 
		\end{align*}

	\noindent
	Applying $\tau \otimes \tau$ to this tensor and using the fact that $\tau$ is a unital trace, the result follows.
	\end{proof}

	By Proposition \ref{prop:ReducedLaplacian}, equation \eqref{eqn:CyclicGradient} 
	may be rewritten

			\begin{equation}
		\label{eqn:OperatorSecondary}
		\tau\bigg{(} (\operator{D}_i + \frac{1}{2}(\operator{Id} \otimes \tau + \tau \otimes \operator{Id})\Delta_i+ \operator{P}_i^V)p\bigg{)} = 0,
	\end{equation}

	\noindent
	where $\operator{P}_i^Vp = (\D_iV)(\D_ip)$.
Summing over $1 \leq i \leq m$, we have

	\begin{equation*}
		\tau\bigg{(} (\operator{D} + \frac{1}{2}\operator{T}_\tau + \operator{P}^V)p\bigg{)} = 0,
	\end{equation*}

\noindent
where 

	\begin{equation*}
		\operator{D} = \sum_{i=1}^m \operator{D}_i, \quad \operator{T}_\tau = (\operator{Id} \otimes \tau + \tau \otimes \operator{Id})\sum_{i=1}^m\Delta_i,
		\quad \operator{P}^V = \sum_{i=1}^m \operator{P}_i^V.
	\end{equation*}

\begin{remark}
	The characteristic property of the operators $\operator{D}=\sum_{i=1}^m \operator{D}_i$ and 
	$\Delta = \sum_{i=1}^m \Delta_i$ is that

		\begin{equation*}
			\tau \otimes \tau(\sum_{i=1}^m \partial_i\D_i p) = \tau(\operator{D}p) + \tau \otimes \tau(\Delta p)
		\end{equation*}

	\noindent
	for any unital trace $\tau$ on $L$.  The transformation $\sum_{i=1}^m \partial_i\D_i$
	is a natural noncommutative analogue of the Laplacian on an $m$-torus.
	The operator $\operator{D}$ is called the 
	\emph{number operator}, and the transformation $\Delta$ is called the 
	\emph{reduced Laplacian}.
\end{remark}

	The summands $\Delta_i$ of the reduced Laplacian $\Delta \in \Hom(L,L^{\otimes2})$
	act on monomials according to the formula \eqref{eqn:ReducedLaplacian}.
	If $p$ is a monomial of degree zero, then the outer sums in this formula are empty, and $\Delta_ip = \mathbf{0} \otimes \mathbf{0}$.
	If $p$ is a monomial of degree one, then the inner sums in this formula are empty, and $\Delta_ip = \mathbf{0} \otimes \mathbf{0}$.
	If $p$ is a monomial of degree $d\geq 2$ which factors as $p=p_1u_ip_2$, and if the cyclic shift $p_2p_1u_i$ factors
	as $p_2p_1u_i=q_1u_iq_2u_i$, then the tensor $q_1u_i \otimes q_2u_i$ has degree at most $d$, but neither
	of its factors has degree zero.  If $p_2p_1u_i$ factors as $p_2p_1u_i=q_1u_i^{-1}q_2u_i$, then the tensor 
	$q_1 \otimes q_2$ has degree at most $d-2$.
	Similarly, if $p$ factors as $p=p_1u_i^{-1}p_2$ and the cyclic shift $u_i^{-1}p_2p_1$ factors as $u_i^{-1}q_1u_iq_2$,
	then the tensor $q_1 \otimes q_2$ has degree at most $d-2$.  If $u_i^{-1}p_2p_1$ factors as $u_i^{-1}q_1u_i^{-1}q_2$,
	then the tensor $u_i^{-1}q_1 \otimes u_i^{-1}q_2$ has degree at most $d$, but neither of its factors has degree zero.
	From these considerations, we conclude that 

		\begin{equation*}
			L_d \xrightarrow{\Delta} \bigvee_{k=1}^{d-1} L_k \otimes L_{d-k}.
		\end{equation*}

	Since $\operator{T}_\tau \in \operatorname{End}L$ is the contraction of 
	$\Delta$ by $\operator{Id} \otimes \tau + \tau \otimes \operator{Id}$,
	we conclude from the above that it is strictly upper triangular with respect to the degree
	filtration in $L$:

		\begin{equation*}
			\dots \xrightarrow{\operator{T}_\tau} L_3  
			\xrightarrow{\operator{T}_\tau} L_2  \xrightarrow{\operator{T}_\tau} L_1  
			\xrightarrow{\operator{T}_\tau}
			\{\mathbf{0}\}.
		\end{equation*}

	The number operator $\operator{D}$ acts diagonally 
	in $L$ with spectrum $0,1,2,\dots$ and corresponding eigenspaces

		\begin{equation*}
			B=L_0, L_1/L_0, L_2/L_1,\dots.
		\end{equation*}

	\noindent
	Obviously, the kernel of $\operator{D}$ is $B$. 

	The operator $\operator{P}^{V}$ is the dot product, $\D V \cdot \D p$, of the cyclic gradient of $V$ with 
	the cyclic gradient of $p$:

		\begin{equation*}
			\operator{P}^{V} p = \sum_{i=1}^m (\D_iV)(\D_ip).
		\end{equation*}

	\noindent
	Unlike $\operator{D}$ and $\operator{T}_\tau$,
	the operator $\operator{P}^{V}$ does not respect the degree filtration in $L$.
	When $V$ is  ``small,'' in an appropriate sense, $\operator{P}^V$ will be a 
	perturbation of the upper triangular operator $\operator{D}+\operator{T}_\tau$,
	hence our notation.

	Since the operators $\operator{D},\operator{T}_\bullet,\operator{P}^\bullet$ 
	all annihilate
	$B$, equation \ref{eqn:OperatorSecondary} contains no information 
	concerning the behaviour of $\tau$ on $B$.  This is an artifact of the 
	cyclic gradient trick, but it results in no loss of information
	since our initial value problem stipulates $\tau|_B = \sigma$.  Now,
	the operator $\operator{D}$ is an automorphism of $B^\perp$, the 
	space of polynomials with no constant term, and hence we may
	regularize by the inverse of this operator.  As we will see in a 
	moment, this regularization has the effect of making the operators
	we have introduced in order to describe the SD equations 
	$\xi$-continuous in the range $\xi>1$.

		\begin{definition}
			For any linear transformation $\operator{T}$
			with domain $L$, we define its \emph{degree
			regularization} by $$\overline{\operator{T}}:=
			\operator{T}\operator{D}^{-1}.$$
			It is understood that the domain of 
			the regularized operator $\overline{\operator{T}}$
			is restricted to $B^\perp$.
		\end{definition}

	We now regularize
	equation \eqref{eqn:OperatorSecondary}, obtaining

		\begin{equation}
			\label{eqn:OperatorNormalized}
			\tau\bigg{(} (\operator{Id} + \frac{1}{2}\overline{\operator{T}}_\tau + \overline{\operator{P}}^V)p\bigg{)} = 0.
		\end{equation}

	\noindent
	A demerit of the operator 

		\begin{equation*}
			\operator{Id} + \frac{1}{2}\overline{\operator{T}}_\tau + \overline{\operator{P}}^V
		\end{equation*}

	\noindent
	is that, since $B^\perp$ is not invariant under the action of 
	the strictly upper triangular operator $\overline{\operator{T}}_\tau$, it
	is not an endomorphism of $B^\perp$.  To rectify this, let
	$\Pi$ be the orthogonal projection of $L$ onto $B^\perp$,
	and let $\Pi'$ be the complementary projection of $L$ on 
	$B$.   

	\begin{definition}
		\label{def:FundamentalOperators}
		Let $\tau$ be a unital trace on $L$, and let $V \in L$ be
		a polynomial.  The \emph{first fundamental operator} associated
		to the data $\tau,V$ is the endomorphism of $B^\perp$ 
		defined by

			\begin{equation*}
				\Psi_\tau^V = \operator{Id} + \frac{1}{2}\Pi\overline{\operator{T}}_\tau + \overline{\operator{P}}^V.
			\end{equation*}

		\noindent
		The \emph{second fundamental operator} associated
		to $\tau,V$ is the endomorphism of $B^\perp$ defined 
		by 

			\begin{equation*}
				\Xi_\tau^V = \operator{Id} + \Pi\overline{\operator{T}}_\tau + 
				\overline{\operator{P}}^V.
			\end{equation*}
	\end{definition}

	Note that the first and second fundamental operators associated to 
	a given unital trace $\tau$ are essentially the same; the precise relation
	between them is 

		\begin{equation*}
			\Xi_\tau^V = \Psi_\tau^V +  \frac{1}{2}\Pi\overline{\operator{T}}_\tau.
		\end{equation*}

	\noindent
	In the next section, we will study a lattice of noncommutative partial differential
	equations, the Schwinger-Dyson lattice, whose rows are described by these operators.
	The first fundamental operator governs the first row of the lattice, while the
	higher rows are controlled by the second fundamental operator.  

	The following property of the fundamental operators 
	follows immediately from their definition.

	\begin{prop}
		\label{prop:FundamentalDistributive}
		For any linear functionals $\tau_0,\dots,\tau_h$, we have

			\begin{equation*}
				\Psi_{\sum_{g=0}^h \tau_g}^V = \Psi_{\tau_0}^V + \sum_{g=1}^h \frac{1}{2}\Pi\overline{\operator{T}}_{\tau_g}
			\end{equation*}

		\noindent
		and 

			\begin{equation*}
				\Xi_{\sum_{g=0}^h \tau_g}^V = \Xi_{\tau_0}^V + \sum_{g=1}^h \Pi\overline{\operator{T}}_{\tau_g}.
			\end{equation*}

	\end{prop}

	In terms of the first fundamental operator, equation \eqref{eqn:OperatorNormalized}
	becomes 

		\begin{equation}
			\label{eqn:OperatorTertiary}
			\tau(\Psi_\tau^V p) = - \frac{1}{2}\sigma(\Pi' \overline{\operator{T}}_\tau p).
		\end{equation}

	\noindent
	Suppose that $\tau,\tau'$ are two solutions of the initial value 
	problem \eqref{eqn:InitialValuePotential}, and set $\delta = \tau'-\tau$.

	\begin{prop}
		\label{prop:MasterDeltaConstraint}
		We have the quadratic constraint

			\begin{equation*}
				\delta(\Xi_\tau^V p) = - \delta \otimes \delta(\overline{\Delta} p), \quad p \in B^\perp.
			\end{equation*}
	\end{prop}

	\begin{proof}
		Since $\tau,\tau'$ are solutions of 
		\eqref{eqn:InitialValuePotential}, we have

			\begin{equation*}
				[\tau' \otimes \tau' - \tau \otimes \tau](\partial_ip) + [\tau'-\tau]((\D_iV)p)=0.
			\end{equation*}

		\noindent
		Using the identity 

			\begin{equation*}
				\tau' \otimes \tau' - \tau \otimes \tau = \delta \otimes \tau + \tau \otimes \delta + \delta \otimes \delta,
			\end{equation*}

		\noindent
		this can be rewritten

			\begin{equation*}
				\delta\bigg{(}(\operator{Id} \otimes \tau + \tau \otimes \operator{Id})\partial_ip \bigg{)} + 
				\delta((\D_iV)p) = -\delta \otimes \delta(\partial_ip).
			\end{equation*}

		\noindent
		Now use the cyclic gradient trick: replace $p$ with $\D_ip$, and 
		sum over $1 \leq i \leq m$.  
	\end{proof}

	\subsection{Operator norm estimates}
	In this subsection, we establish basic continuity properties of the 
	regularized upper triangular operator $\overline{\operator{T}}_\bullet$, the 
	regularized perturbation $\overline{\operator{P}}^{\bullet}$, and
	the regularized reduced Laplacian $\overline{\Delta}$.
	These continuity properties will be essential in the analysis to
	follow.

		\begin{prop}
			\label{prop:contraction} 
			Let $\operator{f} \in \Hom_1(L,\field{C})$.
			Then $\overline{\operator{T}}_f \in \End_\xi(B^\perp)$
			for any $\xi>1$, and

				\begin{equation*}
				 	\|\overline{\operator{T}}_f\|_\xi < 4\|f\|_1\frac{\xi+1}{\xi(\xi-1)}.
				\end{equation*}
		\end{prop}

		\begin{proof}
			By definition, the unregularized operator $\operator{T}_f$ acts on monomials $p$ according to

		\begin{align*}
					&\operator{T}_{\operator{f}} p = \sum_{i=1}^m 
					\sum_{p=p_1u_ip_2} \bigg{(}  \sum_{p_2p_1u_i=q_1u_iq_2u_i} 
						(q_1u_i f(q_2u_i) + f(q_1u_i)q_2u_i) \\&
						- \sum_{p_2p_1u_i=q_1u_i^{-1}q_2u_i} (q_1f(q_2) + f(q_1)q_2)\bigg{)}\\
						&- \sum_{i=1}^m\sum_{p=p_1u_i^{-1}p_2} \bigg{(}\sum_{u_i^{-1}p_2p_1=u_i^{-1}q_1u_iq_2} (q_1 f(q_2) + f(q_1)q_2)\\
						&-
						\sum_{u_i^{-1}p_2p_1=u_i^{-1}q_1u_i^{-1}q_2} 
					(u_i^{-1}q_1 f(u_i^{-1}q_2) + f(u_i^{-1}q_1)u_i^{-1}q_2)
						\bigg{)}.\\
		\end{align*}

	Using the triangle inequality in $(L,\|\cdot\|_\xi)$ and that $|f(p)|\leq \|f\|_1\|p\|_1=\|f\|_1$ for 
	all monomials $p \in L$, we obtain

		\begin{align*}
			&\|\operator{T}_{\operator{f}} p\|_\xi \leq 
			   \|\operator{f}\|_1\sum_{i=1}^m \sum_{p=p_1u_ip_2} \bigg{(}  \sum_{p_2p_1u_i=q_1u_iq_2u_i} ( \|q_1u_i\|_\xi + \|q_2u_i\|_\xi)
			+ \sum_{p_2p_1u_i=q_1u_i^{-1}q_2u_i} (\|q_1\|_\xi + \|q_2\|_\xi) \bigg{)}\\
			&+\|\operator{f}\|_1\sum_{i=1}^m \sum_{p=p_1u_i^{-1}p_2} \bigg{(} \sum_{u_i^{-1}p_2p_1=u_i^{-1}q_1u_iq_2} (\|q_1\|_\xi + \|q_2\|_\xi)
			 + \sum_{u_i^{-1}p_2p_1=u_i^{-1}q_1u_i^{-1}q_2} (\| u_i^{-1}q_1\|_\xi + \|u_i^{-1}q_2\|_\xi)\bigg{)}.
		\end{align*}

	\noindent
	Since $\xi>1$, we have, for $\deg(p)=d$,

				\begin{align*}
					\sum_{p_2p_1u_i=q_1u_iq_2u_i} \|q_1u_i\|_\xi &\leq \xi + \dots + \xi^{d-1} < \frac{1}{\xi-1}\|p\| \\
					\sum_{p_2p_1u_i=q_1u_i^{-1}q_2u_i} \|q_1\|_\xi & \leq 1+\dots+\xi^{d-2} < \frac{1}{\xi(\xi-1)} \|p\|\\
					\sum_{u_i^{-1}p_2p_1=u_i^{-1}q_1u_iq_2} \|q_1\|_\xi  & \leq 1+\dots+\xi^{d-2}< \frac{1}{\xi(\xi-1)} \|p\| \\
					  \sum_{u_i^{-1}p_2p_1=u_i^{-1}q_1u_i^{-1}q_2} \|u_i^{-1}q_1\|_\xi &\leq \xi + \dots + \xi^{d-1} < \frac{1}{\xi-1}\|p\|,
				\end{align*}

	\noindent
	and similarly for the four additional sums involving the symbol $q_2$.
	Thus, we have

		\begin{equation*}
			\|\operator{T}_{\operator{f}} p\|_\xi < 4\|f\|_1\frac{\xi + 1}{\xi(\xi-1)} \deg(p)\|p\|_\xi,
		\end{equation*}

	\noindent
	from which the claim follows.	\end{proof}

	\begin{prop}
		\label{prop:perturbation}
		For any $V \in L$ and any $\xi \geq 1$, we have
		$\overline{\operator{P}}^{V} \in \End_\xi(B^\perp)$ 
		and 

			\begin{equation*} 
				\|\overline{\operator{P}}^{V}\|_\xi \leq \|\Pi V\|_1 \deg(V) \xi^{\deg(V)}.
			\end{equation*}
	\end{prop}

	\begin{proof}
		The operator $\overline{\operator{P}}^{V}$ acts on monomials $p \in B^\perp$ according to

		\begin{equation*}
			\overline{\operator{P}}^{V}p = \frac{1}{\deg p} \sum_{i=1}^m (\D_iV)(\D_ip).
		\end{equation*}

	\noindent
	Thus

			\begin{align*}
				\|\overline{\operator{P}}^{V}p\|_\xi &\leq \frac{1}{\deg p}\sum_{i=1}^m \|(\D_iV)(\D_ip)\|_\xi \leq \frac{1}{\deg p}\sum_{i=1}^m \|\D_iV\|_\xi \|\D_ip\|_\xi \\
				&\leq \frac{1}{\deg p}\bigg{(} \sum_{i=1}^m \|\D_iV\|_\xi \bigg{)} \bigg{(} \sum_{i=1}^m \|\D_ip\|_\xi \bigg{)}.
			\end{align*}

	\noindent
	Since $p$ is a monomial, we have

			\begin{align*}
				\|\D_ip\|_\xi &= \bigg{\|} \sum_{p=p_1u_ip_2} p_2p_1u_i - \sum_{p=p_1u_i^{-1}p_2} u_i^{-1}p_2p_1 \bigg{\|}_\xi \\
				&\leq \sum_{p=p_1u_ip_2} \|p_2p_1u_i\|_\xi + \sum_{p=p_1u_i^{-1}p_2} \|u_i^{-1}p_2p_1\|_\xi \\
				&\leq (\deg_i^+(p) + \deg_i^-(p)) \|p\|_\xi = \deg_i(p) \|p\|_\xi,
			\end{align*}

		\noindent
		so that

			\begin{equation*}
				\sum_{i=1}^m \|\D_ip\|_\xi \leq \deg (p)\|p\|_\xi.
			\end{equation*}

		\noindent
		To estimate the factor depending on $V$, we proceed as follows:

			\begin{align*}
				\sum_{i=1}^m \|\D_iV\|_\xi &= \sum_{i=1}^m \bigg{\|} \D_i \sum_q \langle q,V \rangle q \bigg{\|}_\xi \\
				 &\leq  \sum_{i=1}^m \sum_{q \in B^\perp} |\langle q,V \rangle| \|\D_iq\|_\xi  \\
				  &=  \sum_{q \in B^\perp} |\langle q,V \rangle| \|q\|_\xi \deg(q) \\
				  &\leq \deg (V) \sum_{q\in B^\perp} |\langle q,V \rangle| \|q\|_\xi \\
				  &= \deg (V) \|\Pi V\|_\xi \\
				  &\leq \deg (V) \xi^{\deg (V)} \|\Pi V\|_1 ,
			\end{align*}

		\noindent
		where the last inequality follows from Proposition \ref{prop:xiInequality}.
		Thus we have proved

			\begin{equation}\label{b1}
				\|\operator{P}^{V}\operator{D}^{-1}p\|_\xi \leq \|\Pi V\|_1 \deg(V) \xi^{\deg(V)} \|p\|_\xi,
			\end{equation}

		\noindent
		from which the claim follows.

	\end{proof}

		\begin{prop}
			\label{prop:Laplacian}
			For any $\xi_1,\xi_2 \geq 1$ such that $\xi_1 \geq 2\xi_2$, 
			the regularized reduced Laplacian $\overline{\Delta}$ is a 
			contractive mapping of $(B^\perp,\|\cdot\|_{\xi_1})$ into 
			$(L^{\otimes2},\|\cdot\|_{\xi_2})$.
		\end{prop}

		\begin{proof}
			Let $p \in B^\perp$ be a monomial of degree $d$.  We have

				\begin{equation*}
					\overline{\Delta}p = \frac{1}{d} \sum_{i=1}^m {\Delta}_i p.
				\end{equation*}

			\noindent
			Now,

				\begin{align*}
					\|{\Delta}_ip\|_{\xi_2} &\leq 
					\sum_{p=p_1u_ip_2} 
						\left( \sum_{p_2p_1u_i=q_1u_iq_2u_i} \|q_1u_i\|_{\xi_2} \|q_2u_i\|_{\xi_2} +
						\sum_{p_2p_1u_i=q_1u_i^{-1}q_2u_i} \|q_1\|_{\xi_2} \|q_2\|_{\xi_2} \right) \\
					+&\sum_{p=p_1u_i^{-1}p_2} 
						\left( \sum_{u_i^{-1}p_2p_1=u_i^{-1}q_1u_iq_2} \|q_1\|_{\xi_2} \|q_2\|_{\xi_2} +
						\sum_{u_i^{-1}p_2p_1=u_i^{-1}q_1u_i^{-1}q_2} \|u_i^{-1}q_1\|_{\xi_2} \|u_i^{-1}q_2\|_{\xi_2} \right) \\
					&= \sum_{p=p_1u_ip_2} 
						\left( (\deg_i^+(p)-1) \xi_2^d + \deg_i^-(p) \xi_2^{d-2} \right) +\sum_{p=p_1u_i^{-1}p_2} 
						\left( \deg_i^+(p) \xi_2^{d-2} + (\deg_i^-(p)-1) \xi_2^d \right) \\
					&\leq \deg_i(p)^2\xi_2^d.
				\end{align*}

			\noindent
			We thus have, since $d\le 2^d\le (\frac{\xi_1}{\xi_2})^d$ for all $d\in\mathbb N$, 

				\begin{align*}
					\|\overline{\Delta}p\|_{\xi_2} &\leq \frac{\xi_2^d}{d} \sum_{i=1}^m (\deg_i(p))^2
					\leq d \xi_2^d \leq \left( \frac{\xi_1}{\xi_2} \right)^d \xi_2^d = \xi_1^d = \|p\|_{\xi_1}.
				\end{align*}
		\end{proof}

	\subsection{Uniqueness}
	Let $\tau$ be a unital trace on $L$ such that $\|\tau\|_1 \leq 1$, and let
	$V \in L$ be a potential.
	From Propositions \ref{prop:contraction} and \ref{prop:perturbation}, 
	we conclude that the fundamental operators associated to $\tau,V$ 
	are $\xi$-continuous endomorphisms of $B^\perp$ whose norms satisfy

		\begin{align*}
			\|\Psi_\tau^V-\rm{Id}\|_\xi &<  2\frac{\xi+1}{\xi(\xi-1)} + \|\Pi V\|_1 \deg(V)\xi^{\deg(V)} \\
			\|\Xi_\tau^V-\rm{Id}\|_\xi &<  4\frac{\xi+1}{\xi(\xi-1)} + \|\Pi V\|_1 \deg(V)\xi^{\deg(V)}.
		\end{align*}

	\noindent
	Consequently, $\Psi_\tau^V$ and $\Xi_\tau^V$ extend uniquely
	to continuous endomorphisms of $\mathcal{B}^\perp_\xi$, the completion 
	of $B^\perp$ in the norm $\|\cdot\|_\xi$.  Now, since $\mathcal{B}^\perp_\xi$ is complete,
	$\mathcal{C}(\mathcal{B}^\perp_\xi)$ is a Banach algebra.  Thus,
	if

		\begin{equation}
			\label{eqn:FundamentalInequality}
			K(\xi,V) := 4\frac{\xi+1}{\xi(\xi-1)} + \|\Pi V\|_1 \deg(V)\xi^{\deg(V)} < 1,
		\end{equation}

	\noindent
	then $\Psi_\tau^V$ and $\Xi_\tau^V$ are continuous 
	automorphisms of $\mathcal{B}^\perp_\xi$ with inverses

		\begin{align*}
			(\Psi_\tau^V)^{-1} &= \sum_{n=0}^{\infty} (-1)^n (\frac{1}{2}\overline{\operator{T}}_\tau + \overline{\operator{P}}^V)^n \\
			(\Xi_\tau^V)^{-1} &= \sum_{n=0}^{\infty} (-1)^n (\overline{\operator{T}}_\tau + \overline{\operator{P}}^V)^n 
		\end{align*} 
with norms bounded by
\begin{equation}\label{boundnorm}
\|(\Psi_\tau^V)^{-1}\|_\xi\le \frac{1}{1-K(\xi,V)},\qquad \|(\Xi_\tau^V)^{-1}\|_\xi\le \frac{1}{1-K(\xi,V)}.\end{equation}
We next show that  the condition $K(\xi,V)$ small enough implies uniqueness.   Let $\tau,\tau'$ be solutions of
	\eqref{eqn:InitialValuePotential} such that $\|\tau\|_1,\|\tau'\|_1 \leq 1$, and 
	set $\delta=\tau'-\tau$.  Then, by Proposition \ref{prop:MasterDeltaConstraint}
	and the invertibility of $\Xi_\tau^V$, we have the identity

		\begin{equation*}
			\delta = -(\delta \otimes \delta)\overline{\Delta}(\Xi_\tau^V)^{-1}
		\end{equation*}

	\noindent
	in $\Hom(\mathcal{B}^\perp_\xi,\field{C})$.  Taking operator norms, we obtain
	the inequality

	\begin{equation*}
		\|\delta\|_\xi \leq \|\delta\|_\xi \| (\operator{Id} \otimes \delta)\overline{\Delta}\|_\xi \| (\Xi_\tau^V)^{-1} \|_\xi.
	\end{equation*}

	\noindent
	If $\|\delta\|_\xi \neq 0$, we may cancel it from both sides of this 
	inequality to obtain

		\begin{equation*}
			1 \leq \| (\operator{Id} \otimes \delta)\overline{\Delta}\|_\xi \|(\Xi_\tau^V)^{-1}\|_\xi.
		\end{equation*}

	\noindent
	Using the fact that $\|\delta\|_1 \leq \|\tau'\|_1 + \|\tau\|_1 \leq 2$,
	we proceed as in the proof of Proposition \ref{prop:contraction}
	and find that

		\begin{equation*}
			\| (\operator{Id} \otimes \delta)\overline{\Delta}\|_\xi \leq 4\frac{\xi+1}{\xi(\xi-1)}.
		\end{equation*}

	\noindent
	Combining this with \eqref{boundnorm},
	we obtain the inequality 

		\begin{equation*}
			1 <  4\frac{\xi+1}{\xi(\xi-1)} \frac{1}{1-K(\xi,V)}.
		\end{equation*}

	\noindent
	Let us choose a particular value $\xi_0$ of $\xi$, large enough
	so that $4\frac{\xi+1}{\xi(\xi-1)}<1$.  For example, choosing
	$\xi_0=12$, we have

		\begin{equation*}
			4\frac{\xi_0+1}{\xi_0(\xi_0-1)} = \frac{13}{33} < \frac{1}{2},
		\end{equation*}

	\noindent
	and 

		\begin{equation*}
			K(\xi_0,V) =K(12,V) = \frac{13}{33} + \|\Pi V\|_1 \deg(V) 12^{\deg(V)}.
		\end{equation*}

	\noindent
	Thus if 

		\begin{equation*}
			\|\Pi V\|_1 \deg(V) 12^{\deg(V)} < \frac{1}{2}-\frac{13}{33} = \frac{7}{66},
		\end{equation*}

	\noindent
	we obtain the fallacious inequality
	inequality	$1<1$.  This proves Theorem \ref{thm:uniqueness}.

\section{The Schwinger-Dyson lattice}
\label{sec:SDlattice}
In this section, we introduce the Schwinger-Dyson lattice over $L$.
The \emph{Schwinger-Dyson lattice} with potential $V$ is a
countable set of noncommutative partial differential equations.
The equations $\operatorname{SD}(k,N)$ in this hierarchy are indexed by two discrete parameters, the \emph{order},
$k$, and the \emph{rank}, $N$.  A solution of the Schwinger-Dyson
lattice with potential $V$ is an array

	\begin{equation*}
				\begin{matrix}
					\W_{11}^V & \W_{12}^V & \dots & \W_{1N}^V & \dots \\
					\W_{21}^V & \W_{22}^V & \dots & \W_{2N}^V & \dots \\
				\vdots & \vdots & \ddots & \vdots \\
					\W_{k1}^V & \W_{k2}^V & \dots & \W_{kN}^V & \dots \\
					\vdots & \vdots & {} & \vdots
				\end{matrix}
	\end{equation*}

	\noindent
	whose elements are symmetric multilinear functionals

			\begin{equation*}
				\W_{kN}^V: \underbrace{L \times \dots \times L}_k \rightarrow \field{C}.
			\end{equation*}

	\noindent
	In order to qualify as a solution of the SD lattice, we
	insist that $$\tilde\W_{1N}^V:=N^{-1}\W_{1N}^V$$ is a unital trace, 
	and that $\W_{kN}^V$ is a connected $k$-correlator for $k \geq 2$. 

	By definition, the first equation in the SD lattice, $\SD(1,N)$, is

	\begin{equation*}
		\W_{1N}^V \otimes \W_{1N}^V(\partial_ip) + N\W_{1N}^V((\D_iV)p) = -\W_{2N}^V(\partial_ip).
	\end{equation*}

\noindent 
Note that this equation is invariant under translations of $V$
by elements of $B$.
The subsequent equations in the hierarchy are obtained
 by repeated application of the \emph{Gibbs rule},

	\begin{equation*}
		\frac{\mathrm{d}}{\mathrm{d}z} \W_{kN}^{V+\frac{z}{N}p_{k+1}}(p_1,\dots,p_k)|_{z=0} 
		= \W_{(k+1)N}^V(p_1,\dots,p_k,p_{k+1}).
	\end{equation*}

In this section we give the equations of the SD hierarchy in
explicit form.  First, we present the SD equations in their
\emph{primary form}, obtained directly from the first equation
and iteration of the Gibbs rule.  We then obtain the \emph{secondary form}
of the SD lattice equations by applying the 
primary form to the coordinates of a cyclic gradient,

	\begin{equation*}
		\D p = (\D_1p,\dots,\D_mp).
	\end{equation*}

\noindent
Strictly speaking, the secondary form is a specialization of the 
primary form since not every $m$-tuple of noncommutative polynomials 
occurs as a cyclic gradient.  However, the secondary form of the 
SD equations has the advantage that it is concisely described by
the fundamental operators $\Psi_{\tilde\W_{1N}^V}^V$ and $\Xi_{\tilde\W_{1N}^V}^V$.

Finally, we introduce the notion of \emph{uniformly bounded solutions} of the 
SD lattice.  Quite simply, a uniformly bounded solution is one whose elements
$\W_{kN}^V$ are multilinear functionals whose norms are bounded independently
of $N$.  We will see that, as soon as $\Xi_{\tilde\W_{1N}^V}^V$ is invertible,
uniformly bounded solutions in fact exhibit polynomial decay in $N$.

\subsection{Primary form of the SD equations}
Select $p_1,\dots,p_k \in L$, and consider the perturbed
first order equation

	\begin{equation*}
		\W_{1N}^{V_z} \otimes \W_{1N}^{V_z}(\partial_ip_1) 
		+N \W_{1N}^{V_z}((\D_iV_z)p_1) = - \W_{2N}^{V_z}(\partial_ip_1),
	\end{equation*}

\noindent
where 

	\begin{equation*}
		V_z = V + \sum_{j=2}^k \frac{z_j}{N} p_j
	\end{equation*}

\noindent
and $k \geq 2$.
Equivalently, the perturbed equation is

	\begin{equation*}
		\W_{1N}^{V_z} \otimes \W_{1N}^{V_z}(\partial_ip_1) 
		+N \W_{1N}^{V_z}((\D_iV)p_1)  = - \sum_{j=2}^k z_j \W_{1N}^{V_z}((\D_ip_j)p_1) - \W_{2N}^{V_z}(\partial_ip_1).
	\end{equation*}

\noindent
Applying the Gibbs rule to the perturbed first order equation $k-1$ times, once for each
of the variables $z_2,\dots,z_k$, will yield the $k$th order SD equation at rank $N$.

To write down $\operatorname{SD}(k,N)$ explicitly,
we start by differentiating the term $\W_{1N}^{V_z} \otimes \W_{1N}^{V_z}(\partial_ip_1)$.
Let $q_1 \otimes q_2$ be a simple tensor in $L^{\otimes2}$, and consider

	\begin{equation*}
		\W_{1N}^{V_z} \otimes \W_{1N}^{V_z}(q_1 \otimes q_2) =
		\W_{1N}^{V_z}(q_1)\W_{1N}^{V_z}(q_2).
	\end{equation*}

\noindent
Differentiating with respect to the parameters $z_2,\dots,z_k$ and 
applying the Gibbs rule yields the sum

	\begin{equation*}
		\sum_{r=1}^k \sum_{\substack{I \subseteq \{2,\dots,k\}\\ |I|=r-1}}
		\W_{rN}^V(q_1 \otimes p_I)\W_{(k+1-r)N}^V(q_2 \otimes p_{\overline{I}})
		= \sum_{I \subseteq \{2,\dots,k\}} \W_{(|I|+1)N}^V(q_1 \otimes p_I) \W_{(k-|I|)N}^V(q_2 \otimes p_{I^c})
	\end{equation*}

\noindent
where the sum on the right is over all subsets of $\{2,\dots,k\}$, including the empty set,
$I^c=\{2,\dots,k\}\backslash I$, and 

	\begin{equation*}
		p_I = \bigotimes_{i \in I} p_i, \quad p_{\overline{I}} = \bigotimes_{i \in \overline{I}} p_i.
	\end{equation*}

\noindent
This may be equivalently written 

	\begin{equation*}
		\sum_{ I \subseteq \{2,\dots,k\} }
		\W_{(|I|+1)N}^V \otimes \W_{(k-|I|)N}^V(q_1 \otimes q_2 \# p_I \otimes p_{I^c}),
	\end{equation*}

\noindent
where we are using the notation $q_1 \otimes q_2 \# T= q_1 \otimes T \otimes q_2$ for any
tensor $T$.

Application of the Gibbs rule to the next term, 
$N \W_{1N}^{V_z}((\D_iV)p_1)$, yields the 
contribution

	\begin{equation*}
		N\W_{kN}^V((\D_iV)p_1,p_2,\dots,p_k).
	\end{equation*}

We now move to the right hand side of the perturbed equation.  Application of 
the Gibbs rule yields the contributions

	\begin{align*}
		&\sum_{j=2}^k \W_{(k-1)N}^V((\D_ip_j)p_1,p_2,\dots,\hat{p}_j,\dots,p_k), \quad \W_{(k+1)N}^V(\partial_ip_1,p_2,\dots,p_k),
	\end{align*}

\noindent
where, in the first contribution, the hat denotes an omitted argument.

\begin{prop}[Primary SD equations]
	\label{prop:SDprimary}
	The $k$th order Schwinger-Dyson equation at rank $N$, $\SD(k,N)$, is

		\begin{align*}
			&\sum_{ I \subseteq \{2,\dots,k\} } \W_{(|I|+1)N}^V \otimes \W_{(k-|I|)N}^V(\partial_ip_1 \# p_I \otimes p_{I^c})
			+ N \W_{kN}^V((\D_iV)p_1,\dots,p_k) = \\
			&-\sum_{j=2}^k \W_{(k-1)N}^V((\D_ip_j)p_1,p_2,\dots,\hat{p}_j,\dots,p_k) - \W_{(k+1)N}^V(\partial_ip_1,p_2,\dots,p_k).
		\end{align*}
\end{prop}

\subsection{Secondary form of the SD equations}
Using the fact that $\tilde\W_{1N}^V$ is a unital trace, we use
the cyclic gradient trick and argue as in 
Section \ref{subsec:CyclicGradientTrick} to obtain the secondary
form of the SD equations.
They are expressed in terms of the first
and second fundamental operators associated to
$\tilde\W_{1N}^V$.

	\begin{prop}[Secondary SD equations]
		\label{prop:SDsecondary}
		For any $p \in B^\perp$, we have

			\begin{equation*}
				\W_{1N}^V(\Psi_{\tilde\W_{1N}^V}^Vp) = - \frac{1}{2N}\W_{1N}^V(\Pi'\overline{\operator{T}}_{\tilde\W_{1N}^V}p)
				- \frac{1}{N}\W_{2N}^V(\overline{\Delta}p).
			\end{equation*}

		For any $k \geq2$ and $p_1,\dots,p_k \in B^\perp$, we have

			\begin{align*}
				&\W_{kN}^V(\Xi_{\tilde\W_{1N}^V}^Vp_1,\dots,p_k) = - \frac{1}{N} 
				\sum_I
				\W_{(|I|+1)N}^V \otimes \W_{(k-|I|)N}^V(\overline{\Delta}p_1 \# p_I \otimes p_{\overline{I}}) \\
				&-\frac{1}{N}\sum_{j=2}^k \W_{(k-1)N}^V(\overline{\operator{P}}^{p_j}p_1,\dots,\hat{p}_j,\dots,p_k)
				-\frac{1}{N}\W_{(k+1)N}^V(\overline{\Delta}p_1,\dots,p_k),
			\end{align*}

		\noindent
		where the first sum on the right is over all proper nonempty subsets $I$ of 
		$\{2,\dots,k\}$.
	\end{prop}
	In other words, the sum over $I$ above is taken over all subsets of $\{2,\dots,k\}$ which 
	are neither the full set $\{2,\dots,k\}$, nor the empty set $\emptyset$.

\subsection{Uniform boundedness and renormalization}
So far, we have considered the SD lattice equations from a purely algebraic perspective.
We now inject a modicum of analytic structure by introducing the notion of
\emph{uniformly bounded solutions} of the SD lattice.

	\begin{definition}
		\label{def:UniformlyBounded}
		A solution $(\W_{kN}^V)_{k,N=1}^{\infty}$ of the SD lattice with 
		potential $V$ is said to be \emph{$\xi$-uniformly bounded} if the
		following conditions hold:

			\begin{enumerate}

				\smallskip
				\item
				$\sup_N \|\tilde\W_{1N}^V\|_1 \leq 1$;

				\smallskip
				\item
				For each $k \geq 2$, $\sup_N \|\W_{kN}^V\|_\xi < \infty$.

			\end{enumerate}
		\end{definition}

	Given a $\xi$-uniformly bounded solution as above, we define its 
	\emph{renormalization} by 

			\begin{equation*}
				\tilde\W_{kN}^V = N^{k-2} \W_{kN}^V, \quad k,N \geq 1.
			\end{equation*}

	\noindent
	In terms of the renormalized functionals $\tilde\W_{kN}^V$, uniform boundedness
	means that we have a sequence $(C_k)_{j=2}^{\infty}$ of positive constants
	(i.e. numbers independent of $N$) such that 

		\begin{equation*}
			\|\tilde\W_{kN}^V\|_\xi \leq C_k N^{k-2}, \quad k\geq 2.
		\end{equation*}

	\noindent
	In fact, the SD equations can be used to substantially improve upon this sequence
	of inequalities at the cost of geometrically dilating the $\xi$-norm

	\begin{thm}
		\label{thm:Improvement}
		Let $(\W_{kN}^V)_{k,N=1}^{\infty}$ be a $\xi$-uniformly bounded solution 
		of the SD lattice with potential $V$, and suppose that the corresponding 
		fundamental operators are continuous automorphisms of $\mathcal{B}^\perp_\xi$.
		Set 

			\begin{equation*}
				\xi_l:= 2^{l-2}\xi, \quad l \geq 2.
			\end{equation*}

		\noindent
		There exists an array $(C_{kl})_{k,l=2}^{\infty}$ of constants such that 

			\begin{equation*}
				\|\tilde\W_{kN}^V\|_{\xi_l} \leq C_{kl}N^{\max(0,k-l)}, \quad k,l \geq 2.
			\end{equation*}
	\end{thm}

	\begin{proof}
		Schematically, the theorem statement can be represented as an
		entrywise inequality between two $\frac{\infty}{2} \times \frac{\infty}{2}$ matrices,

			\begin{equation*}
				\begin{bmatrix}
					\|\tilde\W_{2N}^V\|_{\xi_2} & \|\tilde\W_{2N}^V\|_{\xi_3} & \|\tilde\W_{2N}^V\|_{\xi_4} & \|\tilde\W_{2N}^V\|_{\xi_5} & \dots \\
					\|\tilde\W_{3N}^V\|_{\xi_2} & \|\tilde\W_{3N}^V\|_{\xi_3} & \|\tilde\W_{3N}^V\|_{\xi_4} & \|\tilde\W_{3N}^V\|_{\xi_5} & \dots \\
					\|\tilde\W_{4N}^V\|_{\xi_2} & \|\tilde\W_{4N}^V\|_{\xi_3} & \|\tilde\W_{4N}^V\|_{\xi_4} & \|\tilde\W_{4N}^V\|_{\xi_5} & \dots \\
					\|\tilde\W_{5N}^V\|_{\xi_2} & \|\tilde\W_{5N}^V\|_{\xi_3} & \|\tilde\W_{5N}^V\|_{\xi_4} & \|\tilde\W_{5N}^V\|_{\xi_5} & \dots \\
					\vdots & \vdots & \vdots & \vdots 
				\end{bmatrix} \leq 
				\begin{bmatrix}
					C_{22}N^0 & C_{23}N^0 & C_{24}N^0 & C_{25}N^0 & \dots \\
					C_{32}N^1 & C_{33}N^0 & C_{34}N^0 & C_{35}N^0 & \dots \\
					C_{42}N^2 & C_{43}N^1 & C_{44}N^0 & C_{45}N^0 & \dots \\
					C_{52}N^3 & C_{53}N^2 & C_{54}N^1 & C_{55}N^0 & \dots \\
					\vdots & \vdots & \vdots & \vdots 
				\end{bmatrix}.
			\end{equation*}

		\noindent
		We present a proof of this inequality by induction on the column parameter, $l$.

		For $l=2$, the desired statement coincides with the definition of 
		$\xi$-uniform boundedness, and the invertibility of the fundamental
		operators is not required.

		For the induction step, fix $m \geq 2$ and suppose that there exists
		an array of constants 

			\begin{equation*}
				\begin{matrix}
				C_{22} & \dots & C_{2m} \\
				C_{32} & \dots & C_{3m} \\
				 \vdots & {} & \vdots \\
				C_{k2} & \dots & C_{km} \\
				 \vdots & {} & \vdots
				\end{matrix}
			\end{equation*}

		\noindent
		such that

			\begin{equation*}
				\|\tilde\W_{kN}^V\|_{\xi_l} \leq C_{kl}N^{\max(0,k-l)}, \quad k \geq 2,\ l=2,\dots,m.
			\end{equation*}

		\noindent
		We will extend this to an array 

			\begin{equation*}
				\begin{matrix}
				C_{22} & \dots & C_{2m} & C_{2(m+1)} \\
				C_{32} & \dots & C_{3m} & C_{3(m+1)} \\
				 \vdots & {} & \vdots & \vdots \\
				C_{k2} & \dots & C_{km} & C_{k(m+1)} \\
				 \vdots & {} & \vdots & \vdots
				\end{matrix}
			\end{equation*}

		\noindent
		such that

			\begin{equation*}
				\|\tilde\W_{kN}^V\|_{\xi_{m+1}} \leq C_{k(m+1)}N^{\max(0,k-(m+1))}, \quad k \geq 2.
			\end{equation*}

		Let us return
		to the secondary form of the SD equations,
		which in terms of the renormalized functionals $\tilde\W_{kN}^V$ 
		becomes

		\begin{align*}
			&\tilde\W_{kN}^V(\Xi_{\tilde\W_{1N}^V}^Vp_1,\dots,p_k) = 
			-\sum_I \tilde\W_{(|I|+1)N}^V \otimes \tilde\W_{(k-|I|)N}^V(\overline{\Delta}p_1 \# p_I \otimes p_{I^c}) \\
			&-\sum_{j=2}^k \tilde\W_{(k-1)N}^V(\overline{P}^{p_j}p_1,\dots,\hat{p}_j,\dots,p_k)
			-\frac{1}{N^2} \tilde\W_{(k+1)N}^V(\overline{\Delta}p_1,\dots,p_k) \\
			&=:S_{kN}^{(1)}(p_1,\dots,p_k) + S_{kN}^{(2)}(p_1,\dots,p_k) + S_{kN}^{(3)}(p_1,\dots,p_k),
		\end{align*}

		\noindent
		valid for all $k \geq 2$.
		We will use the induction hypothesis to estimate the $\xi_{m+1}$-norm of
		the three contributions $S_{kN}^{(1)},S_{kN}^{(2)},S_{kN}^{(3)}$.

		We begin with $S_{kN}^{(1)}$.  Recall that the summation in this group of terms is over
		proper nonempty subsets $I$ of $\{2,\dots,k\}$.
		We have

			\begin{align*}
				\|S_{kN}^{(1)}\|_{\xi_{m+1}} & \leq \sum_I \|(\tilde\W_{(|I|+1)N}^V \otimes \tilde\W_{(k-|I|)N}^V)\overline{\Delta}\|_{\xi_{m+1}} \\
				&= \sum_{r=1}^{k-2} {k-2\choose r} \|(\tilde\W_{(r+1)N}^V \otimes \tilde\W_{(k-r)N}^V)\overline{\Delta}\|_{\xi_{m+1}} \\
				&\leq  \sum_{r=1}^{k-2} {k-2\choose r} \|\tilde\W_{(r+1)N}^V\|_{\xi_m} \|\tilde\W_{(k-r)N}^V\|_{\xi_m} \|\overline{\Delta}\|_{\xi_{m+1},\xi_m} \\
				&\leq   \sum_{r=1}^{k-2} {k-2\choose r} \|\tilde\W_{(r+1)N}^V\|_{\xi_m} \|\tilde\W_{(k-r)N}^V\|_{\xi_m}.
			\end{align*}

		\noindent
		The second to last inequality follows from the diagram 

			\begin{equation*}
			((B^\perp)^{\otimes k},\|\cdot\|_{\xi_{m+1}})
			\xrightarrow{\overline{\Delta} \otimes \operator{Id}^{\otimes(k-1)}}
			((B^\perp)^{\otimes(k+1)},\|\cdot\|_{\xi_m}) \xrightarrow{\tilde\W_{(r+1)N}^V \otimes \tilde\W_{(k-r)N}^V} \field{C} ,
			\end{equation*}

		\noindent
		and the final inequality is  
		Proposition \ref{prop:Laplacian}.  We 
		now invoke the induction hypothesis, obtaining

			\begin{align*}
				\|S_{kN}^{(1)}\|_{\xi_{m+1}} & \leq  \sum_{r=1}^{k-2} {k-2\choose r} \|\tilde\W_{(r+1)N}^V\|_{\xi_m} \|\tilde\W_{(k-r)N}^V\|_{\xi_m} \\
				&\leq \sum_{r=1}^{k-2} {k-2\choose r} C_{(r+1)m}N^{\max(0,r+1-m)} C_{(k-r)m}N^{\max(0,k-r-m)} \\
				&\leq \sum_{r=1}^{k-2} {k-2\choose r} C_{(r+1)m} C_{(k-r)m}N^{\max(0,r+1-m)+\max(0,k-r-m)} \\
				& \leq C_{k(m+1)}^{(1)} N^{\max(0,k-(m+1))},
			\end{align*}

		\noindent
		where $
				C_{k(m+1)}^{(1)} = \sum_{r=1}^{k-2} {k-2\choose r} C_{(r+1)m} C_{(k-r)m}.$
Next, we estimate the contribution $S_{kN}^{(2)}$.  We have 

			\begin{align*}
				\|S_{kN}^{(2)}\|_{\xi_{m+1}} &\leq \sum_{j=2}^k \|\tilde\W_{(k-1)N}^V\operator{X}_j \|_{\xi_{m+1}} \\
				& \leq \sum_{j=2}^k \|\tilde\W_{(k-1)N}^V\|_{\xi_m} \|\operator{X}_j \|_{\xi_{m+1},\xi_m},
			\end{align*}

		\noindent
		where $\operator{X}_j \in \Hom((B^\perp)^{\otimes k},(B^\perp)^{\otimes(k-1)})$ is the map
		which operators on simple tensors $p_1 \otimes \dots \otimes p_k$ according to 

			\begin{equation*}
				\operator{X}_j p_1 \otimes \dots \otimes p_k = (\overline{\operator{P}}^{p_j}p_1) \otimes \dots \otimes \hat{p}_j \otimes \dots \otimes p_k,
			\end{equation*}

		\noindent
		and the second inequality is the diagram 

			\begin{equation*}
			((B^\perp)^{\otimes k},\|\cdot\|_{\xi_{m+1}})
			\xrightarrow{\operator{X}_j}
			((B^\perp)^{\otimes(k-1)},\|\cdot\|_{\xi_m}) \xrightarrow{\tilde\W_{(k-1)N}^V} \field{C}.
			\end{equation*}

		\noindent
		Now, we claim that the first arrow in this diagram is a contractive mapping.  Indeed,
		for any monomial $p_1 \otimes \dots \otimes p_k \in (B^\perp)^{\otimes k}$, we have

			\begin{align*}
				\|\operator{X}_jp_1 \otimes \dots \otimes p_k\|_{\xi_m} 
				&= \|(\overline{P}^{p_j}p_1) \otimes \dots \otimes \hat{p}_j \otimes \dots \otimes p_k\|_{\xi_m} \\
				&=  \|\overline{P}^{p_j}p_1\|_{\xi_m} \|p_2 \otimes \dots \otimes \hat{p}_j \otimes \dots \otimes p_k\|_{\xi_m} \\
				& \leq \|p_j\|_1 \deg(p_1) \xi_m^{\deg(p_1)} \|p_2 \otimes \dots \otimes \hat{p}_j \otimes \dots \otimes p_k\|_{\xi_m} \\
				&\leq 2^{\deg(p_1)}\xi_m^{\deg(p_1)}  \|p_2 \otimes \dots \otimes \hat{p}_j \otimes \dots \otimes p_k\|_{\xi_m} \\
				&= \|p_1\|_{\xi_{m+1}}  \|p_2 \otimes \dots \otimes \hat{p}_j \otimes \dots \otimes p_k\|_{\xi_m} \\
				& \leq \|p_1 \otimes \dots \otimes p_k\|_{\xi_{m+1}},
			\end{align*}

		\noindent
		where Proposition \ref{prop:perturbation} was applied to obtain the first inequality.
		Thus, returning to our estimate on the $\xi_{m+1}$-norm and applying the induction
		hypothesis, we have

			\begin{align*}
				\|S_{kN}^{(2)}\|_{\xi_{m+1}} &\leq \sum_{j=2}^k \|\tilde\W_{(k-1)N}^V \|_{\xi_m} \\
				& \leq \sum_{j=2}^k C_{(k-1)m} N^{\max(0,k-1-m)}  \leq C_{k(m+1)}^{(2)} N^{\max(0,k-(m+1))},
			\end{align*}

		\noindent
		where $
				C_{k(m+1)}^{(2)} = \sum_{j=2}^k C_{(k-1)m}$.
Finally, we estimate the contribution $S_{kN}^{(3)}$ in $\xi_{m+1}$-norm. 
		From the diagram 

			\begin{equation*}
				((B^\perp)^{\otimes k},\|\cdot\|_{\xi_{m+1}})
				\xrightarrow{\overline{\Delta}}
				((B^\perp)^{\otimes(k+1)},\|\cdot\|_{\xi_m}) \xrightarrow{\tilde\W_{(k+1)N}^V} \field{C},
			\end{equation*}

		\noindent
		Proposition \ref{prop:Laplacian}, and the induction hypothesis, we have the estimate

			\begin{align*}
				\|S_{kN}^{(3)}\|_{\xi_{m+1}} &\leq \frac{1}{N^2} \|\tilde\W_{(k+1)N}^V\overline{\Delta}\|_{\xi_{m+1}} \\
				& \leq \frac{1}{N^2}  \|\tilde\W_{(k+1)N}^V\|_{\xi_m} \|\overline{\Delta}\|_{\xi_m.\xi_{m+1}} \\
				& \leq \frac{1}{N^2}  \|\tilde\W_{(k+1)N}^V\|_{\xi_m} \\
				& \leq  \frac{1}{N^2} C_{(k+1)m} N^{\max(0,k+1-m)} \\
				&= C_{(k+1)m} N^{\max(0,k-1-m)}.
			\end{align*}

		\noindent
		We thus have 

			\begin{equation*}
				\|S_{kN}^{(3)}\|_{\xi_{m+1}} \leq C_{k(m+1)}^{(3)} N^{\max(0,k-(m+1))},
			\end{equation*}

		\noindent
		with $C_{k(m+1)}^{(3)} = C_{(k+1)m}$.
We have now shown that 

			\begin{equation*}
				\|\tilde\W_{kN}^V\|_{\xi_{m+1}} \leq C_{k(m+1)} N^{\max(0,k-(m+1))}, \quad k \geq 2, N \geq 1,
			\end{equation*}

		\noindent
		where $C_{k(m+1)} = C_{k(m+1)}^{(1)} + C_{k(m+1)}^{(2)} + C_{k(m+1)}^{(3)}$, 
		provided that the domain of $\tilde\W_{kN}^V$ is restricted to 

			\begin{equation*}
				\Xi_{\tilde\W_{1N}^V}^V(B^\perp) \otimes (B^\perp)^{\otimes(k-1)}.
			\end{equation*}

		\noindent
		Since our hypotheses dictate that the fundamental operator $\Xi_{\tilde\W_{1N}^V}^V$ is invertible and bounded in 
		$\mathcal{B}^\perp_\xi$, the proof is complete.
	\end{proof}

	\begin{cor}
		\label{cor:LimitPoints}
		Under the hypotheses of Theorem \ref{thm:Improvement},
		the sequence $(\tilde\W_{kN})_{N=1}^{\infty}$ has
		a nonempty set of limit points as a linear form from  ${\mathcal L}_{\xi_k}$ into $\mathbb C$   for any $k \in [1,K]$.
	\end{cor}

	\begin{proof}
		For $k=1$, we have that $\|\tilde\W_{1N}^V\|_1 \leq 1$
		for all $N \geq 1$, so that

			\begin{equation*}
				|\tilde\W_{1N}^V(p)| \leq \|p\|_1 =1
			\end{equation*}

		\noindent
		for any monomial $p \in L$.  Thus, the sequence
		$(\tilde\W_{1N}^V(p))_{N=1}^\infty$ is a bounded
		sequence of complex numbers, and hence admits
		a limit point.  Using the countability of the monomial
		basis in $L$ together with a diagonalization argument
		yields the existence of a limit point of $(\tilde\W_{1N}^V)_{N=1}^{\infty}$
		with respect to the topology of pointwise convergence.

		For $k \geq 2$, we have that $\|\tilde\W_{kN}^V\|_{\xi_k} \leq C_{kk}$ for
		all $N \geq 1$.  Thus, for any monomial $p_1 \otimes \dots \otimes p_k \in L^{\otimes k}$,
		we have

			\begin{equation*}
				|\tilde\W_{kN}^V(p_1,\dots,p_k)| \leq C_{kk}\xi_k^{\deg(p_1)+\dots+\deg(p_k)}
			\end{equation*}

		\noindent
		for all $N\geq 1$, so that $(\tilde\W_{kN}^V(p_1,\dots,p_k))_{N=1}^{\infty}$ is 
		a bounded sequence of complex numbers.  The same countability/diagonalization argument 
		now applies to deduce the existence of a limit point of 
		$(\tilde\W_{kN}^V)_{N=1}^{\infty}$.
	\end{proof}

\section{Asymptotic analysis of the SD equations}
\label{sec:Asymptotics}
Let $V \in L$ be a polynomial verifying 

	\begin{equation}\label{hypV}
		\|\Pi V\|_1 < \frac{7}{66} \cdot \frac{1}{\deg(V) (2^{K-1}12)^{\deg(V)}}.
	\end{equation}

\noindent
Set $\xi=12$ and $\xi_l = 2^{l-2}\xi$ for $l \geq 2$.
Let

	\begin{equation*}
				\begin{matrix}
					\W_{11}^V & \W_{12}^V & \dots & \W_{1N}^V & \dots \\
					\W_{21}^V & \W_{22}^V & \dots & \W_{2N}^V & \dots \\
				\vdots & \vdots & \ddots & \vdots \\
					\W_{k1}^V & \W_{k2}^V & \dots & \W_{kN}^V & \dots \\
					\vdots & \vdots & {} & \vdots
				\end{matrix}
	\end{equation*}

\noindent 
be a $\xi$-uniformly bounded solution of the SD lattice with potential $V$, and let

	\begin{equation*}
				\begin{matrix}
					\tilde{\W}_{11}^V & \tilde{\W}_{12}^V & \dots & \tilde{\W}_{1N}^V & \dots \\
					\tilde{\W}_{21}^V & \tilde{\W}_{22}^V & \dots & \tilde{\W}_{2N}^V & \dots \\
				\vdots & \vdots & \ddots & \vdots \\
					\tilde{\W}_{k1}^V & \tilde{\W}_{k2}^V & \dots & \tilde{\W}_{kN}^V & \dots \\
					\vdots & \vdots & {} & \vdots
				\end{matrix}
	\end{equation*}

\noindent 
be its renormalization.  
Suppose that $\tilde\W_{1N}^V$ restricted to $B$ admits
an $N \rightarrow \infty$ asymptotic expansion to $h$ terms on the asymptotic 
scale $N^{-2}$:

	\begin{equation*}
		\tilde\W_{1N}^V(b) = \sum_{g=0}^h \frac{\sigma_g(b)}{N^{2g}} + o\bigg{(} \frac{1}{N^{2h}} \bigg{)}, \quad b \in B.
	\end{equation*}

\noindent
Under these hypotheses, we prove the following abstract version of Theorem \ref{thm:main}.

\begin{thm}
	\label{thm:asymptotics}
		For each $k \in [1,K]$ , each $h\le K-1$, and all $p_1,\dots,p_k \in B^\perp$,
		the functional $\tilde\W_{kN}^V$ admits an $N \rightarrow \infty$ 
		asymptotic expansion to $h$ terms:

			\begin{equation*}
				\tilde\W_{kN}^V(p_1,\dots,p_k) = \sum_{g=0}^h \frac{\tau_{kg}^V(p_1,\dots,p_k)}{N^{2g}} + o\bigg{(} \frac{1}{N^{2h}}\bigg{)}
			\end{equation*}

		\noindent
		The expansion coefficients $\tau_{kg}^V$ may be described as follows: 

			\begin{enumerate}

				\smallskip
				\item
				$\tau_{10}^V$ is the unique solution of the noncommutative initial value
				problem \eqref{eqn:InitialValuePotential} with $\sigma=\sigma_0$; 

				\smallskip
				\item
				For $k=1$ and $g >0$, 

				$$\tau^V_{1g}(p)=-\sum_{\ell=1}^{g-1}\tau_{1\ell}\otimes\tau_{1 (g-\ell)}(\overline{\Delta}\Xi_{\tau_{10}^V}^{-1} p)
				-\tau_{2 (g-1)}(\overline{\Delta}\Xi_{\tau_{10}^V}^{-1});$$

				\smallskip
				\item
				For $k >1$ and $g>0$,

				\begin{align*}
					&\tau_{kg}^V(p_1,\dots,p_k) = 
					-\sum_{f=1}^g \tau_{k(g-f)}^V(\overline{\operator{T}}_{\tau_{1f}^V}(\Xi^V_{\tau_{10}^V})^{-1}p_1,\dots,p_k) \\
					&\quad - \sum_I \sum_{f=0}^g 
					\tau_{(|I|+1)f}^V \otimes \tau_{(k-|I|)(g-f)}^V(\overline{\Delta}(\Xi^V_{\tau_{10}^V})^{-1}p_1 \# p_I \otimes p_{\overline{I}}) \\
					&\quad -\sum_{j=2}^k \tau_{(k-1)g}^V(\overline{\operator{P}}^{p_j}(\Xi_{\tau_{10}^V}^V)^{-1}p_1,\dots,\hat{p}_j,\dots,p_k) 
					-\tau_{(k+1)(g-1)}^V(\overline{\Delta}(\Xi_{\tau_{10}^V}^V)^{-1}p_1,\dots,p_k),
				\end{align*}

				\noindent
				where the first sum on the right is over all proper nonempty subsets $I$ of $\{2,\dots,k\}$. 

			\end{enumerate}
\end{thm}

\begin{remark}
	We reiterate that the asymptotics of $\tilde\W_{1N}^V$ when restricted to $B$ are 
	part of our hypotheses, while for $k \geq 2$ we need only consider 
	$\tilde\W_{kN}^V(p_1,\dots,p_k)$ for arguments with no constant terms, by
	multilinearity and connectedness.  This is why Theorem \ref{thm:asymptotics}
	is stated for $p_1,\dots,p_k \in B^\perp$.
\end{remark}

Theorem \ref{thm:asymptotics} is proved using Theorem \ref{thm:uniqueness} together 
with Theorem \ref{thm:Improvement} and further analysis of the renormalized SD equations

	\begin{equation}
			\begin{split}
			&\label{eqn:SDRenormalized}
			\tilde\W_{kN}^V(\Xi_{\tilde\W_{1N}^V}^Vp_1,\dots,p_k) =
			- \sum_I
			\tilde\W_{(|I|+1)N}^V \otimes \tilde\W_{(k-|I|)N}^V(\overline{\Delta}p_1 \# p_I \otimes p_{\overline{I}}) \\
			&-\sum_{j=2}^k \tilde\W_{(k-1)N}^V(\overline{\operator{P}}^{p_j}p_1,\dots,\hat{p}_j,\dots,p_k)
			- \frac{1}{N^2}\tilde\W_{(k+1)N}^V(\overline{\Delta}p_1,\dots,p_k).
			\end{split}
		\end{equation}

\noindent
Using equation \eqref{eqn:SDRenormalized}, we establish theorem \ref{thm:asymptotics}
through a double induction on the parameters $k$ and $g$.
As will be clear from the proof, the error terms are uniform on potentials $V$ satisfying \eqref{hypV}.
In particular our argument allows that the potential $V$ may itself depend on $N$,
for example be given by

	\begin{equation*}
		V=V_0 + \frac{1}{N}V_1
	\end{equation*}

\noindent
for each $N \geq 1$, with $V_0,V_1 \in L$ fixed. In this case,   $\tau^V_{kg}$ depends on $N$ through $V$,
and can be again expanded recursively. 
We will make use of this fact in Section \ref{sec:MatrixModels}.

\subsection{External base step: $g=0$}
Here we will prove that 

	\begin{equation*}
		\lim_{N \rightarrow \infty} \tilde\W_{kN}^V(p_1,\dots,p_k) = \tau_{k0}^V(p_1,\dots,p_k)
	\end{equation*}

\noindent
for all $k$ such that $K(\xi_{k+2},V)<1$
 and $p_1,\dots,p_k \in B^\perp$, with 
$\tau_{k0}^V$ as given in the statement of Theorem \ref{thm:asymptotics}.  
The proof is by induction on $k$.

	\subsubsection{Internal base step: $k=1$}
	Let $\lim \tilde\W_{1N}^V$ denote the set of limit points 
	of the sequence $(\tilde\W_{1N}^V)_{N=1}^{\infty}$.  
	By Corollary \ref{cor:LimitPoints}, this is a nonempty set.

	Let $\tau \in \lim \tilde\W_{1N}^V$.  Then, there is a subsequence
	$(\tilde\W_{1N_n}^V)_{n=1}^{\infty}$ of $(\tilde\W_{1N}^V)_{N=1}^{\infty}$ converging pointwise
	to $\tau$ on $L$.  From the initial form of $\SD(1,N_n)$, we have that

		\begin{equation*}
			\tilde\W_{1N_n}^V \otimes \tilde\W_{1N_n}^V(\partial_ip) + \tilde\W_{1N_n}^V((\D_iV)p) 
			=-\frac{1}{N_n^2} \tilde\W_{2N_n}^V(\partial_ip)
		\end{equation*}

	\noindent
	for all $n \geq 1$.  By Theorem \ref{thm:Improvement}, we have 

		\begin{equation*}
			|\tilde\W_{2N_n}(\partial_ip)| \leq \|\tilde\W_{2N_n}\|_{\xi_2} \|\partial_ip\|_{\xi_2}
			\leq C_{22} N^{\max(0,2-N_n)}\deg (p) \|p\|_{\xi_2}
		\end{equation*}

	\noindent
	for all $n \geq 1$.  Taking the $n \rightarrow \infty$ limit
	we obtain that for $p\in {\mathcal L}_{\xi_2}$, 

		\begin{equation*}
			\tau \otimes \tau(\partial_ip) + \tau((\D_iV)p) = 0.
		\end{equation*}

	\noindent
	Moreover, $\tau|_B=\sigma_0$, by hypothesis.  Thus
	$\tau$ is a solution of the initial value problem 
	\eqref{eqn:InitialValuePotential} with $\sigma=\sigma_0$.
	By Theorem \ref{thm:uniqueness}, this initial value 
	problem admits at most one solution if $K(\xi,V)<1$.  Consequently,
	$\lim\tilde\W_{1N}^V$ consists of a single point, and
	this is $\tau_{10}^V$.

	\subsubsection{Internal induction step: $k>1$}
	Fix $k \geq2$, and suppose that $\tilde\W_{rN}^V$ converges
	to $\tau_{r0}^V$  as a linear form on ${\mathcal L}_{\xi_r}^{\otimes r}$ for all $1 \leq r < k$, with $\tau_{r0}^V$ as 
	specified in Theorem \ref{thm:asymptotics}. Assume also $K(\xi_{k+2},V)<1$ so that $\Xi_{\tau_{10}^V}^V$ is invertible in ${\mathcal L}_{\xi_{k+2}}$ by \eqref{boundnorm}.

	By Corollary \ref{cor:LimitPoints}, the sequence $(\tilde\W_{kN}^V)_{N=1}^{\infty}$
	has a nonempty set of limit points.  Let $\tau \in \lim \tilde\W_{kN}^V$ be a limit 
	point.  To prove that $\tau=\tau_{k0}^V$, we return to the renormalized form of
	$\SD(k,N)$ given in equation \ref{eqn:SDRenormalized}.  Let $(\tilde\W_{kN_n}^V)_{n=1}^{\infty}$
	be a subsequence converging to $\tau$.  We then have that

		\begin{equation*}
			\begin{split}
			&\tilde\W_{kN_n}^V(\Xi_{\tilde\W_{1N}^V}^Vp_1,\dots,p_k) =
			- \sum_I
			\tilde\W_{(|I|+1)N_n}^V \otimes \tilde\W_{(k-|I|)N_n}^V(\overline{\Delta}p_1 \# p_I \otimes p_{\overline{I}}) \\
			&-\sum_{j=2}^k \tilde\W_{(k-1)N_n}^V(\overline{\operator{P}}^{p_j}p_1,\dots,\hat{p}_j,\dots,p_k)
			- \frac{1}{N_n^2}\tilde\W_{(k+1)N_n}^V(\overline{\Delta}p_1,\dots,p_k)
			\end{split}
		\end{equation*}

	\noindent 
	for all $n \geq 1$.  By Proposition \ref{prop:Laplacian} and Theorem \ref{thm:Improvement},
	we have that 

		\begin{align*}
			|\tilde\W_{(k+1)N_n}^V(\overline{\Delta}p_1,\dots,p_k)| &\leq \|\tilde\W_{(k+1)N_n}^V\|_{\xi_{k+1}}
			\|\overline{\Delta}\|_{\xi_{k+2},\xi_{k+1}}\|p_1 \otimes \dots \otimes p_k\|_{\xi_{k+2}} \\
			&\leq C_{(k+1)(k+1)}\|p_1 \otimes \dots \otimes p_k\|_{\xi_{k+2}}.
		\end{align*}

	\noindent
	Hence, by the induction hypothesis, we have for $p_i\in {\mathcal L}_{\xi_{k+2}}$, 

		\begin{equation*}
			\begin{split}
			&\tau(\Xi_{\tau_{10}^V}p_1,\dots,p_k)=
			- \sum_I
			\tau_{(|I|+1)0}^V \otimes \tau_{(k-|I|)0}^V(\overline{\Delta}p_1 \# p_I \otimes p_{\overline{I}}) \\
			&-\sum_{j=2}^k \tau_{(k-1)0}^V(\overline{\operator{P}}^{p_j}p_1,\dots,\hat{p}_j,\dots,p_k).
			\end{split}
		\end{equation*}

	\noindent
	Since $\Xi_{\tau_{10}^V}^V$ is invertible in ${\mathcal L}_{\xi_{k+2}}$,  $\tau$ is uniquely defined and we obtain $\tau=\tau_{k0}^V$, as required.

	\subsection{The second column: $g=1$}

	In this section we obtain the second term in the asymptotics
	of $\tilde{\W}_{kN}^V$, i.e. the limit of the error functional 

		\begin{equation*}
			\delta_1\tilde{\W}_{kN}^V(p_1,\dots,p_k) = \tilde{\W}_{kN}^V(p_1,\dots,p_k) - \tau_{k0}^V(p_1,\dots,p_k).
		\end{equation*}

	\noindent
	As in the previous section, our argument is inductive in $k$.

	\subsubsection{Internal base step: $k=0$}

	Our starting point is the following quadratic constraint,
	which is the analogue of Proposition \ref{prop:MasterDeltaConstraint}
	for the error functional.

		\begin{prop}
			\label{prop:ErrorConstraint}
			For any $p \in B^\perp$, we have

			\begin{equation*}
				\delta_1\tilde{\W}_{1N}^V(\Xi_{\tau_{10}^V}^V p) = 
				- \delta_1\tilde{\W}_{1N}^V \otimes \delta_1\tilde{\W}_{1N}^V(\overline{\Delta}p)
				- \frac{1}{N^2}\tilde{\W}_{2N}^V(\overline{\Delta}p).
			\end{equation*}
		\end{prop}

		\begin{proof}
			The proof is analogous to the proof of Proposition \ref{prop:MasterDeltaConstraint}.
			We start with the equations

				\begin{align*}
					\tilde\W_{1N}^V \otimes \tilde\W_{1N}^V(\partial_ip) + \tilde\W_{1N}^V((\D_iV)p) &= - \frac{1}{N^2}\tilde\W_{2N}^V(\partial_ip) \\
					\tau_{10}^V \otimes \tau_{10}^V(\partial_ip) + \tau_{10}^V((\D_iV)p) &= 0.
				\end{align*}

			\noindent
			Subtracting the second equation from the first yields the identity

				\begin{equation*}
					[\tilde\W_{1N}^V \otimes \tilde\W_{1N}^V - \tau_{10}^V \otimes \tau_{10}^V](\partial_ip) 
					+ [\tilde\W_{1N}^V - \tau_{10}^V]((\D_iV)p) = - \frac{1}{N^2}\tilde\W_{2N}^V(\partial_ip).
				\end{equation*}
which is equivalent  to

				\begin{equation*}
					\delta_1\tilde\W_{1N}^V \bigg{(}(\operator{Id} \otimes \tau_{10}^V + \tau_{10}^V \otimes \operator{Id})\partial_ip \bigg{)}
					+\delta_1\tilde\W_{1N}^V((\D_iV)p)
					= -  \delta_1\tilde\W_{1N}^V \otimes \delta_1\tilde\W_{1N}^V(\partial_ip) - \frac{1}{N^2}\tilde\W_{2N}^V(\partial_ip).
				\end{equation*}

			\noindent
			Now use the cyclic gradient trick.
		\end{proof}

	We now use Proposition \ref{prop:ErrorConstraint} to obtain an upper 
	bound on the error functional, showing in particular that it is bounded on the
	correct asymptotic scale.

		\begin{prop}
			\label{prop:RoughEstimate}
			For any $N \geq 1$, we have

				\begin{equation*}
					\|\delta_1\tilde\W_{1N}^V\|_{\xi_3} \leq 
					\frac{C_{22}\|(\Xi_{\tau_{10}^V}^V)^{-1}\|_{\xi_3}}{1-4\frac{\xi_3+1}{\xi_3(\xi_3-1)}\|(\Xi_{\tau_{10}^V}^V)^{-1}\|_{\xi_3}} \cdot \frac{1}{N^2}
				\end{equation*}
				which is finite as soon as $K(\xi_3,V)<1$ by \eqref{boundnorm}.
		\end{prop}

		\begin{proof}
		Since $\Xi_{\tau_{10}^V}^V$ is continuous and invertible, the error constraint implies the identity 

			\begin{eqnarray}
				\delta_1\tilde{\W}_{1N}^V &=& - \delta_1\tilde{\W}_{1N}^V \otimes \delta_1\tilde{\W}_{1N}^V\overline{\Delta}(\operator{\Xi}_{\tau_{10}^V}^V)^{-1}
					- \frac{1}{N^2}\tilde{\W}_{2N}^V\overline{\Delta}(\Xi_{\tau_{10}^V}^V)^{-1} \nonumber\\
					&=& -\delta_1\tilde{\W}_{1N}^V (\operator{Id} \otimes \delta_1\tilde{\W}_{1N}^V)\overline{\Delta} (\Xi_{\tau_{10}^V}^V)^{-1}
					- \frac{1}{N^2}\tilde{\W}_{2N}^V\overline{\Delta} (\Xi_{\tau_{10}^V}^V)^{-1}\label{toto1}
			\end{eqnarray}

		\noindent
		in $\Hom(\mathcal{B}^\perp_{\xi},\field{C})$.  
		Since $|\delta_1\tilde{\W}_{1N}^V(p)| \leq 2$ for all monomials $p$, the same argument 
		as in the proof of Proposition \ref{prop:contraction} yields the bound 

			\begin{equation*}
				\|(\operator{Id} \otimes \delta_1\tilde{\W}_{1N}^V) \overline{\Delta}\|_{\xi_3} \leq 4\frac{\xi_3+1}{\xi_3(\xi_3-1)}.
			\end{equation*}

		\noindent
		Furthermore, we have

			\begin{equation*}
				\|\tilde\W_{2N}^V\|_{\xi_3} \leq \|\tilde\W_{2N}^V\|_{\xi_2} \|\overline{\Delta}\|_{\xi_3,\xi_2} \leq C_{22},
			\end{equation*}

		\noindent
		by Proposition \ref{prop:Laplacian} and Theorem \ref{thm:Improvement}.
		We thus have the inequality

			\begin{equation*}
				\|\delta_1\tilde{\W}_{1N}^V\|_{\xi_3} \leq \|\delta_1\tilde{\W}_{1N}^V\|_{\xi_3}
				4\frac{\xi_3+1}{\xi_3(\xi_3-1)}\|(\Xi_{\tau_{10}^V}^V)^{-1}\|_{\xi_3}
				+\frac{C_{22}}{N^2}\|(\Xi_{\tau_{10}^V}^V)^{-1}\|_{\xi_3},
			\end{equation*}
from which the result follows. 	\end{proof}

	We can now complete the proof of the base step.  
	Consider the set $\lim N^2\delta_1\tilde\W_{1N}^V$ of limit 
	points of the sequence $(N^2\delta_1\tilde\W_{1N}^V)_{N=1}^{\infty}$.  
	By Proposition \ref{prop:RoughEstimate},
	this set is nonempty, just as in the proof of Corollary \ref{cor:LimitPoints}.  
	Moreover, by Proposition \ref{prop:ErrorConstraint},
	any limit point $\tau$ must satisfy the equation

		\begin{equation*}
			\tau(\Xi_{\tau_{10}^V}^Vp) = -\tau_{20}^V(\overline{\Delta}p).
		\end{equation*}

	\noindent
	Since $\Xi_{\tau_{10}^V}$ is invertible, we conclude that 
	$N^2\delta_1\tilde\W_{1N}^V$ converges to $\tau_{11}^V$ given by

		\begin{equation*}
			\tau_{11}^V(p) =  -\tau_{20}^V(\overline{\Delta}(\Xi_{\tau_{10}^V}^V)^{-1}p),
		\end{equation*}

	\noindent
	as required.

	\subsubsection{Internal induction step: $k \geq 2$}
	Fix $k \geq 2$, and suppose that

		\begin{equation*}
			\lim_{N \rightarrow \infty} N^2\delta_1\tilde\W_{rN}^V = \tau_{r1}^V
		\end{equation*}

	\noindent
	for all $1 \leq r < k$, with $\tau_{r0}^V$ and $\tau_{r1}^V$ as 
	in Theorem \ref{thm:asymptotics}. Assume $K(\xi_{k+2}, V)<1$.

	As in the induction step for the first column (i.e. $g=0$), we return to the renormalized
	equation $\SD(k,N)$:

		\begin{equation*}
			\begin{split}
			&\label{eqn:SDrenormalized}
			\tilde\W_{kN}^V(\Xi_{\tilde\W_{1N}^V}p_1,\dots,p_k) =
			- \sum_I
			\tilde\W_{(|I|+1)N}^V \otimes \tilde\W_{(k-|I|)N}^V(\overline{\Delta}p_1 \# p_I \otimes p_{\overline{I}}) \\
			&-\sum_{j=2}^k \tilde\W_{(k-1)N}^V(\overline{\operator{P}}^{p_j}p_1,\dots,\hat{p}_j,\dots,p_k)
			- \frac{1}{N^2}\tilde\W_{(k+1)N}^V(\overline{\Delta}p_1,\dots,p_k).
			\end{split}
		\end{equation*}

	\noindent
	Expanding $\tilde\W_{1N}^V$ to two terms and applying Proposition \ref{prop:FundamentalDistributive},
	the left hand side is

		\begin{align*}
			\text{LHS} &= 
			\tilde\W_{kN}^V(\Xi_{\tau_{10}^V+\frac{\tau_{11}^V}{N^2}}^Vp_1,\dots,p_k) + o\left(\frac{1}{N^2} \right) \\
			&= \tilde\W_{kN}^V(\Xi_{\tau_{10}^V}^Vp_1,\dots,p_k) 
			+ \frac{1}{N^2}\tilde\W_{kN}^V(\overline{\operator{T}}_{\tau_{11}^V}p_1,\dots,p_k) + o\left( \frac{1}{N^2} \right) \\
			&= \tilde\W_{kN}^V(\Xi_{\tau_{10}^V}^Vp_1,\dots,p_k) 
			+ \frac{1}{N^2}\tau_{k0}^V(\overline{\operator{T}}_{\tau_{11}^V}p_1,\dots,p_k) + o\left( \frac{1}{N^2} \right),
		\end{align*}

	\noindent
	where we used the external base step in the last line.  Similarly, applying the induction hypothesis and the 
	external base step, the right hand side becomes

		\begin{align*}
			\text{RHS} &= \tau_{k0}^V(p_1,\dots,p_k) \\ 
					-& \frac{1}{N^2} \sum_I [\tau_{(|I|+1)0}^V \otimes \tau_{(k-|I|)1}^V + 
						\tau_{(|I|+1)1}^V \otimes \tau_{(k-|I|)0}^V](\overline{\Delta}p_1 \# p_I \otimes p_{\overline{I}}) \\
					-& \frac{1}{N^2}\sum_{j=2}^k(\overline{\operator{P}}^{p_j}p_1,\dots,\hat{p}_j,\dots,p_k) 
					- \frac{1}{N^2}\tau_{(k+1)0}^V(\overline{\Delta}p_1,\dots,p_k) 
					 + o\bigg{(} \frac{1}{N^2} \bigg{)}.
		\end{align*}

	\noindent
	Putting these together and using the fact that $\Xi_{\tau_{10}^V}^V$ is invertible,
	we conclude that 

		\begin{equation*}
			\tilde\W_{kN}^V(p_1,\dots,p_k) = \tau_{k0}^V(p_1,\dots,p_k) + \frac{\tau_{k1}^V(p_1,\dots,p_k)}{N^2} + o\bigg{(} \frac{1}{N^2} \bigg{)},
		\end{equation*}

	\noindent
	where 

		\begin{align*}
			\tau_{k1}^V(p_1,\dots,p_r) &= - \tau_{k0}^V(\overline{\operator{T}}_{\tau_{11}^V}(\Xi_{\tau_{10}^V}^V)^{-1}p_1,\dots,p_k) \\
			&-\sum_I [\tau_{(|I|+1)0}^V \otimes \tau_{(k-|I|)1}^V + 
				\tau_{(|I|+1)1}^V \otimes \tau_{(k-|I|)0}^V](\overline{\Delta}(\Xi_{\tau_{10}^V}^V)^{-1}p_1 \# p_I \otimes p_{\overline{I}}) \\
				&-\sum_{j=2}^k(\overline{\operator{P}}^{p_j}(\Xi_{\tau_{10}^V}^V)^{-1}p_1,\dots,\hat{p}_j,\dots,p_k) - \tau_{(k+1)0}^V(\overline{\Delta}(\Xi_{\tau_{10}^V}^V)^{-1}p_1,\dots,p_k),
		\end{align*}

	\noindent
	as required.

	\subsection{External induction step: $g \geq 2$}
	For each $g$ in the range $0 \leq g \leq h$, define

		\begin{equation*}
			\delta_g\tilde\W_{kN}^V = \tilde\W_{kN}^V - \left(\tau_{k0}^V + \dots + \frac{\tau_{k(g-1)}^V}{N^{2(g-1)}} \right),
			\quad k \geq 1,
		\end{equation*}

	\noindent
	where by convention $\delta_0\tilde\W_{kN}^V = \tilde\W_{kN}^V$.
	Our goal is to prove that 

		\begin{equation*}
			\lim_{N \rightarrow \infty} N^{2g}\delta_g\tilde\W_{kN}^V = \tau_{kg}^V
		\end{equation*}

	\noindent
	for each $k \geq 1$ so that $K(\xi_{k+2},V)<1$ and $g\le K= max\{k: K(\xi_{k+2},V)<1\}-1$,
	with $\tau_{kg}^V$ as in Theorem \ref{thm:asymptotics}.  So far we have shown this for $g=0,1$.

	We now fix $2 \leq g \leq h$, and suppose that 

		\begin{equation*}
			\lim_{N \rightarrow \infty} N^{2f}\delta_f\tilde\W_{kN}^V = \tau_{kf}^V
		\end{equation*}

	\noindent
	for each $0 \leq f < g\le K-1$ and all $k \in [1,K]$.  

	\subsubsection{Internal base step: $k=1$}
	Using Proposition \ref{prop:ErrorConstraint}, writing for $i=1,2$
	$$\delta_1\tilde W^V_{iN}=\sum_{f=1}^{g-1}\frac{1}{N^{2f}}\tau_{if}+\frac{1}{N^{2(g-1)}}\delta_{g}\tilde W^V_{iN}$$
	with $\delta_{g}\tilde W^V_{iN}$ going to zero as $N$ goes to infinity,
	and identifying each orders in $N^{-2f}, 1\le f\le g-1$, 
	we arrive at 
	$$N^2 \delta_{g} \tilde{\mathcal  W}^V_{1N}(\Xi_{\tau_{10}^V}p)=-\sum_{f=1}^{g-1}\tau_{1f}\otimes\tau_{1(g-f)} 
	(\overline{\Delta} p)-\tau_{2 (g-1)}(\overline{\Delta} p)+o(1)$$
	from which we deduce that the sequence $N^2  \delta_{g} \tilde{\mathcal W}^V_{1N}$ converges 
	to $\tau_{1g}^V$ as claimed in Theorem \ref{thm:asymptotics}. Note here that the $g$th term in the expansion of $\tilde {\mathcal W}^V_{1n}$ depends
	on the $g-1$ first terms in the expansion of $\tilde{\mathcal  W}^V_{2n}$. We will soon see that  the latter itself will depend on the $g-2$ first terms in the
	expansion of $\tilde \W^V_{3n}$... so that ultimately the $g$th term in the expansion  of $\tilde\W^V_{1N}$ will depend on the first term in the expansion of $\tilde W^V_{(g+1) N}$. This is the reason why we can obtain the expansion only up to order $K-1$.

	\subsubsection{Internal induction step: $k \geq 2$}
	Let $k \geq 2$, and suppose that $\tilde\W_{rN}^V$ admits
	the expansion 

		\begin{equation*}
			\tilde\W_{rN}^V(p_1,\dots,p_k) = \sum_{f=0}^g \frac{\tau_{rf}^V(p_1,\dots,p_r)}{N^{2f}} + o\bigg{(} \frac{1}{N^{2g}} \bigg{)}
		\end{equation*}

	\noindent
	for all $1 \leq r < k$, with the expansion coefficients $\tau_{rf}^V$
	as given in Theorem \ref{thm:asymptotics}.

	To complete the induction, we must prove the claimed expansion for $\tilde\W_{kN}^V$.
	As above, we return to the renormalized SD equation,

		\begin{equation*}
			\begin{split}
			&\label{eqn:SDrenormalized}
			\tilde\W_{kN}^V(\Xi_{\tilde\W_{1N}^V}p_1,\dots,p_k) =
			- \sum_I
			\tilde\W_{(|I|+1)N}^V \otimes \tilde\W_{(k-|I|)N}^V(\overline{\Delta}p_1 \# p_I \otimes p_{\overline{I}}) \\
			&-\sum_{j=2}^k \tilde\W_{(k-1)N}^V(\overline{\operator{P}}^{p_j}p_1,\dots,\hat{p}_j,\dots,p_k)
			- \frac{1}{N^2}\tilde\W_{(k+1)N}^V(\overline{\Delta}p_1,\dots,p_k),
			\end{split}
		\end{equation*}

	\noindent
	By the induction hypothesis and Proposition \ref{prop:FundamentalDistributive},
	the left hand side of the SD equation expands as 

		\begin{align*}
			\text{LHS} &= \tilde\W_{kN}^V(\Xi_{\tilde\W_{1N}^V}p_1,\dots,p_k)  \\
			&=  \tilde\W_{kN}^V \bigg{(} (\Xi_{\tau_{10}^V} + \sum_{f=1}^{g} \frac{1}{N^{2f}} \overline{\operator{T}}_{\tau_{1f}^V} )p_1,\dots,p_k \bigg{)}
				+ o\bigg{(} \frac{1}{N^{2g}} \bigg{)} \\
			&=  \tilde\W_{kN}^V(\Xi_{\tau_{10}^V}p_1,\dots,p_k)+ 
				\sum_{f=1}^{g} \frac{1}{N^{2f}} \tilde\W_{kN}^V(\overline{\operator{T}}_{\tau_{1f}^V} p_1,\dots,p_k)  + o\bigg{(} \frac{1}{N^{2g}} \bigg{)}\\
			&=  \tilde\W_{kN}^V(\Xi_{\tau_{10}^V}p_1,\dots,p_k)+ 
				\sum_{f=1}^{g} \frac{1}{N^{2f}} \bigg{(}\sum_{e=0}^{g-1} \frac{1}{N^{2e}} \tau_{ke}^V(\overline{\operator{T}}_{\tau_{1f}^V} p_1,\dots,p_k) 
				+ o\bigg{(} \frac{1}{N^{2(g-1)}}\bigg{)}\bigg{)}  \\
			&=  \tilde\W_{kN}^V(\Xi_{\tau_{10}^V}p_1,\dots,p_k)+ 
				 \sum_{e=0}^{g-1} \sum_{f=1}^{g} \frac{1}{N^{2(e+f)}} \tau_{ke}^V(\overline{\operator{T}}_{\tau_{1f}^V} p_1,\dots,p_k) 
				 + o\bigg{(} \frac{1}{N^{2g}} \bigg{)},
		\end{align*}

	\noindent
	so that the term of order $N^{-2g}$ is 

		\begin{equation*}
			\sum_{f=1}^g \tau_{k(g-f)}^V(\overline{\operator{T}}_{\tau_{1f}^V}p_1,\dots,p_k).
		\end{equation*}

	\noindent
	Similarly, the term of order $N^{-2g}$ on the right hand side is

		\begin{align*}
			&- \sum_I \sum_{f=0}^g 
			\tau_{(|I|+1)f}^V \otimes \tau_{(k-|I|)(g-f)}^V(\overline{\Delta}\Xi_{\tau_{10}^V}^{-1}p_1 \# p_I \otimes p_{\overline{I}}) \\
			&-\sum_{j=2}^k \tau_{(k-1)g}^V(\overline{\operator{P}}^{p_j}\Xi_{\tau_{10}^V}^{-1}p_1,\dots,\hat{p}_j,\dots,p_k) 
			-\tau_{(k+1)(g-1)}^V(\overline{\Delta}\Xi_{\tau_{10}^V}^{-1}p_1,\dots,p_k).
		\end{align*}

	\noindent
	Putting these two together, we obtain the expansion 

		\begin{equation*}
			\tilde\W_{kN}^V(p_1,\dots,p_k) = \tau_{k0}^V(p_1,\dots,p_k) + \dots + \frac{\tau_{kg}^V(p_1,\dots,p_k)}{N^{2g}} + o\bigg{(} \frac{1}{N^{2g}}\bigg{)},
		\end{equation*}

	\noindent
	with the expansion coefficients as claimed in Theorem \ref{thm:asymptotics}.

\section{Matrix model solutions of the SD lattice}
\label{sec:MatrixModels}

	In this section, we treat matrix model solutions 
	of the SD lattice.  Select $V_0,V_1,\ldots, V_k \in L$ and set 

		\begin{equation*}
			V = \sum_{\ell=0}^k \frac{1}{N^\ell}  V_\ell, \quad N \geq 1.
		\end{equation*}

	\noindent
	Let 

		\begin{equation*}
			\rho_N:B \rightarrow \operatorname{Mat}_N(\field{C}), \quad N \geq 1,
		\end{equation*}

	\noindent
	be a sequence of matrix representations of $B$, and consider the Gibbs 
	ensemble generated by $\rho_N(V)$.  We recall that this is the sequence
	of complex, unit-mass Borel measures $\mu_N^V$ defined by the density 

		\begin{equation*}
			\mu_N^V(\mathrm{d}\mathbf{U}) = \frac{1}{Z_N^V} e^{N\Tr \rho_N(V)} \mu_N(\mathrm{d}\mathbf{U}),
		\end{equation*}

	\noindent
	where $\mu_N$ is Haar measure on the compact group $U(N)^m$.  

	For each $p \in L$, set

		\begin{equation*}
			\W_{1N}^V(p) = \int \Tr\rho_N(p)(\mathbf{U}) \mu_N^V(\mathrm{d}\mathbf{U}), \quad N \geq 1.
		\end{equation*}

	\noindent 
	It is immediate from the form of the density and the moment-cumulant formula
	that the higher cumulants may be obtained by iterating the Gibbs rule,

		\begin{equation*}
			\frac{\mathrm{d}}{\mathrm{d}z} \W_{kN}^{V+\frac{z}{N}p_{k+1}}(p_1,\dots,p_k)|_{z=0} = \W_{(k+1)N}^V(p_1,\dots,p_k).
		\end{equation*}

	\noindent
	Thus, by construction, the cumulants 
		\begin{equation*}
				\begin{matrix}
					\W_{11}^V & \W_{12}^V & \dots & \W_{1N}^V & \dots \\
					\W_{21}^V & \W_{22}^V & \dots & \W_{2N}^V & \dots \\
					\vdots & \vdots & \ddots & \vdots \\
					\W_{k1}^V & \W_{k2}^V & \dots & \W_{kN}^V & \dots \\
					\vdots & \vdots & {} & \vdots
				\end{matrix}
		\end{equation*}

	\noindent
	form a solution of the SD lattice with potential $V$.  

	The condition 

		\begin{equation*}
			|\W_{1N}^V(p)| \leq N
		\end{equation*}

	\noindent 
	for all monomials $p \in L$ is immediate: the matrix $\rho_N(p)(\mathbf{U})$ is a product of 
	unitary matrices and contractions.  We thus have $\|\tilde\W_{1N}^V\|_1 \leq 1$.  However,
	the existence of $\xi \geq 1$ such that this solution is $\xi$-uniformly bounded
	not automatic.  In this section, we use concentration of measure
	techniques or change of variables tricks  
	to verify that uniform boundedness holds
	for the cumulants of \emph{real} Gibbs ensembles.

	\subsection{Concentration of measure}
	Suppose that $\rho_N(V)$ generates a real Gibbs ensemble $\mu_N^V$.
	Let $\mathbf{U}_N^V =(U_1,\dots,U_m)$ be an $m$-tuple of $N \times N$ random
	unitary matrices whose joint distribution in $U(N)^m$ is $\mu_N^V$.

	In this subsection, we assume that $V$ is both selfadjoint and balanced, as in Definition \ref{def:balanced}. 
	Then, for any $\theta_1,\ldots,\theta_k\in [0,2\pi)$ and any 
	$U_1,\dots,U_m \in U(N)^m$, we have
	$$\tr \rho_N(V)(e^{i\theta_1} U_1,\ldots,e^{i\theta_m}U_m)=\tr \rho_N(V)(U_1,\ldots,U_m)\,.$$

		\begin{lem}
			\label{concentration}
			Suppose that $V=\sum \beta_i q_i$ with $\sum |\beta_i|\deg(q_i)^2$ small enough.
			Then, there exist  constants $c,C>0$ such that for any monomial $q\in L$,  for any $\delta\ge 0$
			$$\mu_N^V\bigg{\{} \mathbf{U} \in U(N)^m : |\tr \rho_N(q)({\bf U})-\mathbb E\tr \rho_N(q)({\bf U}_N^V)| \ge \delta   \deg (q) \bigg{\}}
			\le C e^{-c\delta^2 }$$
		\end{lem}

	\begin{proof}
	Let us first remark that, since $V$ is balanced, for any balanced monomial $q$
	the law of $$X^N_q:=\frac{1}{N}\tr \rho_N(q)({\bf U}_N^V)$$ under $\mathbb \mu_N^V$ is the same 
	as the law under $\tilde{\mu}_N^V$, the law defined under the Haar measure on $SU(N)$ and with the same density.
	If $q$ is not balanced, notice that its expectation under $\mu_N^V$  vanishes as it has the same law
	as  $e^{i\theta_k} X^N_q$ with some $\theta_k$ independent of $X^N_q$ and taken uniformly on $[0,2\pi)$ (see details in the proof of 
	\cite[Corollary 4.4.31]{AGZ}). Moreover,  the law of $|X^N_q|$ under under $\mu_N^V$ is the same 
	as its law under $\tilde{\mu}_N^V$. Hence, it is enough to prove concentration inequalities under $\tilde{\mu}_N^V$.
	But we know, by a result of Gromov, see \cite[Theorem 4.4.27]{AGZ}, that the Ricci curvature of $SU(N)$ is bounded below by
	$2^{-1}(N+2)-1$. On the other hand, the Hessian of the function

		\begin{equation*}
			\mathbf{U} \mapsto  \tr \rho_N(V)(\mathbf{U})
		\end{equation*}

	\noindent
	is bounded above by
	$\sum | \beta_i| \deg(q_i)^2$. Hence, by the Bakry-Emery criterion, see \cite[Corollary 4.4.25]{AGZ}, we know that if $0<c=2^{-1}-\sum | \beta_i| \deg(q_i)^2$,
	we have 

	\begin{equation}
		\label{toto}
		\tilde{\mu}_N^V\left(|G-\int G d\tilde{\mu}_N^V|\ge\delta\right)\le 2 e^{-\frac{ (cN-1) \delta^2}{2 \|\Gamma_1(G)\|_\infty}}\,,
	\end{equation}

	\noindent
	for all measurable functions $G$ on $SU(N)$,
	where $\Gamma_1$ is the {\it  carr\'e du champ}. On the other hand it is well known (see e.g. \cite[p. 75]{G} 
	for a similar argument on $SO(N)$) that the metric on $SU(N)$ can be lower bounded
	by the Euclidean metric on the full set of matrices. In particular,  if $G=\frac{1}{N}\tr \rho_N(q)({\bf U})$, 
	we can bound $\|\Gamma_1(G)\|_\infty$ from above by  noticing that

	\begin{eqnarray*}
	\left|\tr \rho_N(q)({\bf U})-\tr\rho_N(q)(\tilde {\bf U})\right|
	&\le& \sum_{\varepsilon=\pm 1\atop 1\le j\le m}\sum_{q=p_1 U_j^\varepsilon p_2}\left|\tr\left(\rho_N(p_1)({\bf U})((U_j)^{\varepsilon}-(\tilde U_j)^{\varepsilon})
	\rho_N(p_2)(\tilde {\bf U})\right)\right|\\
	&\le& \sum_{j=1}^m \deg_j(q) \tr\left|U_j -\tilde U_j\right|\\
	&\le& {\sqrt{N}} (\sum_{j=1}^m \deg_j(q)^2)^{\frac{1}{2}} \left(\sum_{j=1}^m \tr (U_j -\tilde U_j)(U_j -\tilde U_j)^*\right)^{\frac{1}{2}}\\
	\end{eqnarray*}

	Hence,
	$$\|\Gamma_1(\tr\rho_N(q)({\bf U}))\|_\infty\le  N\sum_{j=1}^m \deg_j(q)^2\,.$$
	Plugging this estimate back into \eqref{toto} completes the proof, as $\sum_{j=1}^m \deg_j(q)^2\le \deg q^2$.

	\end{proof}
\subsection{Uniform boundedness via change of variables}\label{conch}
	In this section we get estimates on correlators by using change of variables and the good controls we
	had on the operator $\Xi_\bullet^\bullet$.  This elaborates on  the change of variable approach introduced in this context in  \cite{CGM}.
	 We assume throughout this section that $V$ is self-adjoint (but not necessarilyy balanced).
	We shall first prove that
	\begin{lem}\label{expcont}
	For any monomials $p_1,\dots,p_m$ such that $p_i=-p_i^*$, $1\le i\le m$ ,  let
	$Z_{p_1,\dots,p_m}:U(N)^m \rightarrow \field{C}$ be the function
	$$Z_{p_1,\ldots,p_m}(\mathbf{U})= \sum_{i=1}^m \left(\frac{1}{N}\tr \otimes \frac{1}{N}\tr \rho_N(\partial_i p_i)(\mathbf{U}) +
	\frac{1}{N}\tr \rho_N(\D_i V p_i)(\mathbf{U})\right)\,.$$
	Then there exists a universal constant $C$ depending only on $V$ such that for
	any $\lambda\in \mathbb R$, 
	$$\int  e^{\lambda N  Z_{p_1,\ldots,p_m}(\mathbf{U})}\mu_N^V(\mathrm{d}\mathbf{U})  \le e^{C\lambda^2(\max_i \deg(p_i))^2}\,.$$
	\end{lem}
	\begin{proof} We start by noting that $Z_{p_1,\ldots,p_m}$ is real when $p_i=-p_i^*$. Indeed 
	we have if $p=-p^*$ and for all $i$
\begin{eqnarray}
{\mathcal D}_ip&=&\sum\langle p, q\rangle {\mathcal D}_i q\nonumber\\
&=& \sum\langle p, q\rangle[\sum_{p=q_1U_iq_2} q_2q_1 U_i-\sum_{p=q_1 U_i^*q_2}U_i^*q_2 q_1]\nonumber\\
{\mathcal D}_ip^*&=& \sum \overline{\langle p, q\rangle}[-\sum_{p=q_1U_iq_2} U_i^* q_1^* q_2^* +\sum_{q=q_1 U_i^*q_2}q_1^* q_2^*U_i]\nonumber\\
&=&-({\mathcal D}_ip)^*\label{selfadjoint}\end{eqnarray}
so that ${\mathcal D}_ip=({\mathcal D}_ip)^*$ if $p^*=-p$, which implies that
$\tr \rho_N(\D_i V p_i)(\mathbf{U})$ is real if $V$ is  self-adjoint and $p_i=-p_i^*$.
Similarly
\begin{eqnarray*}
\partial_i p&=&\sum\langle p, q\rangle\sum_{\epsilon=\pm 1}\epsilon \sum_{p=q_1u_i^\epsilon q_2} q_1u_i^{1_{\epsilon=+1}}\otimes u_i^{1_{\epsilon=-1}}q_2\\
\partial_i p^*&=&-\sum\langle p, q\rangle\sum_{\epsilon=\pm 1}\epsilon \sum_{p=q_1u_i^\epsilon q_2}  q_2^*u_i^{1_{\epsilon=-1}}\otimes u_i^{1_{\epsilon=1}}q_1^*\\
\end{eqnarray*}
implies that 
$$\frac{1}{N}\tr \otimes \frac{1}{N}\tr \rho_N(\partial_i p)(\mathbf{U})=-\overline { \frac{1}{N}\tr \otimes \frac{1}{N}\tr \rho_N(\partial_i p^*)(\mathbf{U})}$$
which 
shows that $(\frac{1}{N}\tr \otimes \frac{1}{N}\tr \rho_N(\partial_i p)(\mathbf{U}) $ is real if $p=-p^*$. Hence $Z_{p_1,\dots,p_m}$ is real when $p_i=-p_i^*$ for all $i\in\{1,\ldots,m\}$.
	
	The proof  of the lemma goes through the change of variables
	$${\bf U}=(U_1,\ldots, U_m)\rightarrow {\bf \Psi(U)} =(\Psi_1(U_1),\ldots,\Psi_m(U_m))$$ with
	$$\Psi_j(U)=U_j e^{\frac{\lambda}{N} \rho_N(p_j)(U_j)}$$
	It was shown in \cite[section 2]{CGM} that for $N$ large enough ${\bf \Psi}$ is a local diffeomorphism
	of $U(N)^m$.
	We thus have
	\begin{equation}\label{eqL}1=\int |\det J_\Psi(\mathbf{U})| e^{N\tr \rho_N(V)({\bf \Psi(U)})-N\tr \rho_N(V)({\bf U})} 
	\mu_N^V(\mathrm{d} \mathbf{U})\end{equation}
	The  Jacobian  of $\Psi$ can be computed by
	$$|\det J_\Psi(U)|=\exp(-\sum_{p\ge 1}\frac {(-\lambda/N)^p}{p}\tr(\Phi(U)^p))\,.$$
	Here $\Phi$ is the linear map acting on $\mathcal A_N=\{A\in M_N: A=-A^*\}$ 
	$$\Phi_{ij}(U).A=\sum_{k=0}^{\infty} \frac{({\rm Ad}_{\frac{\lambda}{N} \rho_N(p_j)}(U))^k}{(k+1)!} (\rho_N(\partial_ip_j)\#A)e^{\frac{\lambda}{N}\rho_N(p_j)(U)}$$
	where ${\rm Ad}_MH=MH-HM$. Moreover it was shown that
	\begin{eqnarray*}
	\tr(\Phi(U))&=& \sum_{i=1}^m[ \tr\otimes \tr\rho_N(\partial_i p_i)+ \sum_{k=1}^{+\infty}\tr\otimes\tr  \frac{({\rm Ad}_{\frac{\lambda}{N} \rho_N(p_j)(U)})^k}{(k+1)!} \rho_N(\partial_ip_i)e^{\frac{\lambda}{N}p_i(U)}\\
	&&+\sum_{k=0}^{+\infty}\tr\otimes \tr  \frac{({\rm Ad}_{\frac{\lambda}{N} p_i(U)})^k}{(k+1)!} (\partial_ip_i)(e^{\frac{\lambda}{N}p_i(U)}-1)]\end{eqnarray*}
	We can estimate the other terms as $\|{\rm Ad}_MH\|_\infty\le 2\|M\|_\infty\|H\|_\infty$ and $$\|\frac{\lambda}{N} \rho_N(p_j)({\bf U})\|_\infty\le \frac{2\lambda}{N}$$
Therefore, for all $k\ge 0$
$$\|\tr\otimes\tr  \frac{({\rm Ad}_{\frac{\lambda}{N} \rho_N(p_j)({\bf U})})^k}{(k+1)!} (\partial_i\rho_N(p_i)({\bf U})e^{\frac{\lambda}{N}\rho_N(p_i)({\bf U})}\|_\infty\le N^2 \frac{(\frac{4\lambda}{N})^k}{(k+1)!} \deg(p_i) $$
$$\|\tr\otimes \tr  \frac{({\rm Ad}_{\frac{\lambda}{N} \rho_N(p_i)({\bf U})})^k}{(k+1)!} (\partial_ip_i)(e^{\frac{\lambda}{N}\rho_N(p_i)({\bf U})}-1)\|_\infty \le N^2 \frac{(\frac{4\lambda}{N})^{k+1}}{(k+1)!} \deg(p_i) $$
so that for  all $|\lambda|\le N$
\begin{equation}\label{b11}
|\tr(\Phi(U))-\sum_{i=1}^m \tr\otimes \tr(\partial_i \rho_N(p_i)({\bf U}))|\le  C\max \deg(p_i) \lambda N\end{equation}
Moreover, for all $i,j$, we have 
$$\|\Phi_{ij}(U)\|_\infty\le \deg(p_j)  \sum_{k\ge 0} \frac{(\frac{4\lambda}{N})^k}{(k+1)!}\le  C \deg(p_j)$$
for some finite constant $C$, 
so that for $\ell\ge 2$ all $|\lambda|\le N$, 
$$\tr(\Phi(U)^\ell)\le  N^2 (C \deg(p))^\ell \,.$$
Thus, for $|\lambda|<N/2C \deg(p)$ we deduce that
$$| \sum_{\ell \ge 2}\frac {(-\lambda/N)^\ell}{\ell}\tr(\Phi(U)^\ell)|\le \sum_{\ell\ge 2}  N^2 (C \max \deg(p_i))^\ell \frac {(2|\lambda|/N)^\ell}{\ell}\le2 \lambda^2(\max\deg(p_i))^2 \,.$$ 
Hence, we deduce  that 
\begin{equation}\label{expdet}
|\det J_\Psi(U)|= e^{-\frac{\lambda}{N}\sum_{i=1}^m \tr\otimes \tr(\partial_i p_i) +X_P^N({\bf U})}\end{equation}
where there exists a universal finite constant $C$ such that
$$\|X^N_P\|_\infty \le   C(\max_i \deg(p_i))^2 \lambda^2 \,.$$
Moreover, one easily check that
$$\tr \rho_N(V)({\bf \Psi(U)})-\tr\rho_N( V)({\bf U})=\frac{\lambda}{N} \sum_{i=1}^m \tr(\rho_N(p_i D_i V)({\bf U}))+\epsilon^N_p$$
where $|\epsilon^N_p|\le C \frac{\lambda^2}{N}$. 
Hence,  \eqref{eqL}  proves the claim. \end{proof}
As a corollary we have
\begin{cor} Let $p$ be a polynomial such that $p=p^*$ and set
	$$Z_{p}= N\sum_{i=1}^m (\frac{1}{N}\tr \otimes \frac{1}{N}\tr \rho_N(\partial_i \D_i p) +\frac{1}{N}\tr \rho_N(\D_i V  \D_i p))\,.$$
	Then there exists a universal constant $C$ depending only on $V$ such that for
	any $\lambda\in \mathbb R$, 
	$$\int e^{\lambda  Z_{p}({\bf U})}\mu^V_N(d{\bf U})  \le e^{C\lambda^2\sum |\langle p,q\rangle| \deg(q)^3}\,.$$
	\end{cor}
\begin{proof} By \eqref{selfadjoint}, $Z_p$ is real. 
Moreover, by H\"older inequality, we have 
\begin{eqnarray*}
\mathbb E[ e^{\lambda Z_p}]&=&\mathbb E[ e^{\lambda \sum_q \langle p,q\rangle Z_q}]\\
&\le& \prod_{q}\mathbb E[ e^{\lambda\mbox{sign}(\langle p, q\rangle) Z_q}]^{|\langle p,q\rangle|}\\
&\le &\prod_{q,i,q'} \mathbb E[ e^{\lambda\mbox{sign}(\langle p, q\rangle)\mbox{ sign}( \langle {\mathcal{D}}_i q, q'-(q')^*\rangle) Z_{0,\ldots, q'-(q')^*,0,\ldots}}]^{|\langle p,q\rangle||\langle {\mathcal{D}}_i q, q'-(q')^*\rangle}\\
&\le& e^{C\lambda^2 \sum_q |\langle p,q\rangle|\sum_{i,q'} |\langle {\mathcal D}_iq, q'-(q')^*\rangle| \deg(q')^2}
\\
&\le& e^{C\lambda^2 \sum_q |\langle p,q\rangle| \deg(q)^3}\end{eqnarray*}
where we used the previous lemma.
\end{proof}

We next use the previous lemma as well as the Schwinger-Dyson equations to prove concentration estimates.
\begin{lem}\label{conconc}
Let $\delta_N(p)=\frac{1}{N}\tr\rho_N(p)(\mathbf{U})-\tau_{10}^V(p)$ and assume that  $\xi\ge 12$ so that $\Xi_{\tau_{10}^V}^V$ is invertible in ${\mathcal L}_{\xi}$. Assume moreover 
that $\|V\|_1$ is small enough so that we can choose $\xi$ so that additionally
$$max_{\mu \in \Char(L), \|\mu\| \leq 1} \|\Pi\operator{T}_{\mu}\operator{D}^{-1}\|_\xi  \|(\Xi_{\tau_{10}^V}^V)^{-1}\|_{\xi}<2\,.$$
Then, for all $r\ge 1$, there exists a finite constant $C_r$ such that for all polynomials such that $\|p\|_{\xi}\le 1$, we have
 $$\mathbb E[|\delta_N(p)|^r ]^{\frac{1}{r}}\le \frac{C_r}{N }
$$
\end{lem}
\begin{proof}
 If we denote by  $L_N=\frac{1}{N}\tr\circ\rho_N$, we can rewrite the definition of $Z_p$ as
$$Z_p=L_N\otimes L_N (\sum \partial_i {\mathcal D}_i p)+L_N(\sum {\mathcal D}_i V{\mathcal D} _ip)=L_N(\Psi_{L_N}\operator{D} p)$$
so that

\begin{equation}\label{truc}
\delta_N(\Xi_{\tau_{10}^V}^V  p)=- \frac{1}{2}\delta_N(\Pi\operator{T}_{\delta_N}\operator{D}^{-1} p) + \frac{1}{N} Z_{\operator{D}^{-1} p}
\end{equation}
so that taking the $\ell^r$-norm on both sides we deduce that for any polynomial $p=p^*$ we have
$$\|\delta_N(p)\|_r \le \frac{1}{2} \|\delta_N ((\Pi\operator{T}_{\delta_N}\operator{D}^{-1}  (\Xi_{\tau_{10}^V}^V)^{-1} p) \|_r +\frac{1}{N} \|Z_{\operator {D}^{-1} (\Xi_{\tau_{10}^V}^V)^{-1} p}\|_r$$
We set $\|\delta_N\|_{r}=\max_{\|p\|_\xi\le 1} \|\delta_N(p)\|_r$
and deduce that
$$\|\delta_N\|_{r}\le \frac{1}{2} \max_{\mu \in \Char(L), \|\mu\|_\xi \leq 1} \|\Pi\operator{T}_{\mu}\operator{D}^{-1}\|_\xi   \|(\Xi_{\tau_{10}^V}^V)^{-1}\|_{\xi}  \|\delta_N\|_{r}+
\frac{1}{N} \max_{\|p\|_\xi\le 1} \|Z_{\operator{D}^{-1} (\Xi_{\tau_{10}^V}^V)^{-1} p}\|_r$$
Remembering  that we chose the norm so that 
$$ \max_{\mu \in \Char(L), \|\mu\| \leq 1} \|\Pi\operator{T}_{\mu}\operator{D}^{-1}\|_\xi  \|(\Xi_{\tau_{10}^V}^V)^{-1}\|_\xi <2$$
the conclusion follows from the fact that the previous lemma yields the existence of a finite constant so that since $\xi$ is large enough so that $n^3\le \xi^n$ for all $n\in\mathbb N$,
\begin{eqnarray*}
 \|Z_{\operator{D}^{-1}( \Xi^V_{\tau^V_{10}})^{-1} p}\|_r&\le & C'_r \sum_q |\langle \operator{D}^{-1} (\Xi^{V}_{\tau_{10}^V})^{-1} p, q\rangle|\deg(q)^3\\
 &\le& \|\operator{D}^{-1} (\Xi_{\tau_{10}^V}^V)^{-1}\|_\xi \|p\|_\xi  <\infty\end{eqnarray*}
\end{proof}

\subsection{Uniform boundedness}

\begin{cor}  \label{concmes}
	Under the assumptions of Lemma \ref{concentration}  or section \ref{conch},
	 for any $k\ge 2$, any  monomials $p_1,\cdots, p_k$ we have
	$$|\W^V_{kN}(p_1,\cdots p_k)|\le C_k \prod \deg (p_i) \|p_i\|_\xi$$
	for a finite constant  $C_k$.  In particular, $(\W_{kN}^V)_{k,N=1}^{\infty}$ is $\xi$-uniformly bounded
	for any $\xi \geq 12$.
\end{cor}

\begin{proof}

	By induction we easily check that

		$$\W^V_{kN}(p_1,\cdots p_k)=\int \prod_{i=1}^k \bigg{(}\Tr \rho_N(p_i)({\bf U})-\mathbb E\Tr\rho_N(p_i)({\bf U}_N)\bigg{)}d\mu_N^V(\mathrm{d}{\bf U})\,.$$

	Hence the result follows by H\"older's inequality. This is trivial in the setting when Lemma \ref{conconc} applies but also when 
	Lemma \ref{concentration} does since it  yields

	\begin{eqnarray*}
		&&\int \left(\Tr \rho_N(p)({\bf U})-\mathbb E\Tr \rho_N(p) ({\bf U})\right)^{2k} \mu_N^V(\mathrm{d}\mathbf{U})\\
		&\le& 2k  \int_0^\infty  x^{2k-1} \mu_N^V\bigg{\{} \mathbf{U} \in U(N)^m : |\tr \rho_N(p)({\bf U})-
			\mathbb E\tr \rho_N(p)({\bf U}_N)|\ge x \bigg{\}}dx\\
		&\le& 2k C \int_0^\infty  x^{2k-1} e^{-c\frac{x^2}{(\deg p)^2}} dx\,.\\
	\end{eqnarray*}

\end{proof}

\section{Consequences of the main result}
\label{sec:consequences}

\subsection{Expansion of the free energy: Proof of Corollary \ref{cor:main}}
The expansion of the free energy is a direct consequence of Theorem \ref{thm:main}  as we can write
$$\frac{1}{N^2}\log Z^N_{V}=\int_0^1\partial_t \frac{1}{N^2}\log Z^N_{tV} dt=\int_0^1 \frac{1}{N} \W^{tV}_{1N}(V) dt$$
where for all $t\in [0,1]$ $tV$ satisfies the hypothesis of Theorem \ref{thm:main} as soon as $V$ does. Hence, the asymptotic expansion for the
free energy is a direct consequence of the asymptotic expansion of $\W_{1N}^{tV}, t\in [0,1]$ which is uniform in $t\in [0,1]$.

	\subsection{Central limit theorem: Proof of Corollary \ref{cor:CLT}}

	We write that
	$$\log \mu_N^V\left( e^{\lambda (\Tr P)}\right)=
	\int_0^\lambda \frac{ \mu_N^V (\Tr P e^{\alpha \Tr P})}{\mu_N^V( e^{\alpha \Tr P})}d\alpha= \int_0^\lambda \mathcal{ W}_{1N}^{V+\frac{\alpha}{N} P}(P)d\alpha$$
	Moreover, by Theorem \ref{thm:asymptotics}, we know that 
	$$\mathcal{W}^{V+\frac{\alpha}{N} P}_{1N}(P)=N\tau_{10}^{V+\frac{\alpha}{N} P}(P)+O(\frac{1}{N})$$
	where it is not hard to check that the error is uniform in $\alpha\in [0,\lambda]$.
	Therefore we need to compute the expansion of $\tau_{10}^{V+N^{-1} V_1}$. It is not difficult to see that
	$\tau_{10}^{V+\epsilon  V_1}$ is smooth in $\epsilon$ so that we can write   $$\tau_{10}^{V+\epsilon V_1}=:\tau_{10}^{V}+\epsilon \tau_{10}^{V,V_1} +o(\epsilon)$$
	Plugging back this expansion into the Schwinger-Dyson equation shows that $ \tau_{10}^{V, V_1}$ is solution of 
	$$\tau_{10}^{V, V_1}(\Xi_{\tau^V_{10}}^V p)=-\tau_{10}^V(\overline{\operator{P}^{V_1}} p)\Rightarrow \tau_{10}^{V, V_1} (p)=
	-\tau_{10}^V(\overline{\operator{P}^{V_1} }(\Xi^V_{\tau^V_{10}})^{-1}p) \,.$$
	Since $\overline{\operator{P}^V}$ is linear in $V$ we find  by taking $\epsilon=N^{-1}$ that
	$$  W^{V+\frac{\alpha}{N} p}(p)=N\tau_{10}^V(p) -\frac{\alpha}{N} \tau_{10}^V(\overline{\operator{P}^{p} }(\Xi^V_{\tau^V_{10}})^{-1}p)+o(\frac{1}{N})$$
	from which we conclude that
	$$\lim_{N\ra\infty} \log  \mu_N^V\left( e^{\lambda (\Tr p- N\tau_{10}^V(p)  )}\right)
	=-\int_0^\lambda \alpha  \tau_{10}^V(\overline{\operator{P}^{p} }(\Xi_{\tau^V_{10}}^V)^{-1}p)d\alpha=-\frac{\lambda^2}{2}  \tau_{10}^V(\overline{\operator{P}^{p} }(\Xi^V_{\tau_{10}^V})^{-1}p).$$

	\subsection{Proof of Theorem \ref{thm:HCIZ}}
	We now give the proof of Theorem \ref{thm:HCIZ} as stated in the introduction.  

	Since the monomial $xuyu^{-1}$ is selfadjoint up to cyclic 
	symmetry, for any $t \in \field{R}$, the quadratic potential
	$V_t=txuyu^{-1}$ generates a real Gibbs ensemble, i.e. the Borel measure $\mu_N^{V_t}$
	on $U(N)$ with density

		\begin{equation*}
			\frac{1}{Z_N^{V_t}} e^{N \Tr \rho_N(V_t)} 
			= \frac{1}{Z_N^{V_t}} e^{tN \Tr \rho_N(x)U\rho_N(y)U^{-1}}
		\end{equation*}

	\noindent
	against the Haar measure $\mu_N$ is a real probability measure.
	Thus, by Corollary \ref{cor:main}, for real $t$ satisfying

		\begin{equation*}
			|t| < \frac{7}{66} \cdot \frac{1}{22^{K-1}\cdot 12^2} = \frac{7}{19008},
		\end{equation*}

	\noindent
	the free energy 

		\begin{equation*}
			F_N(t) = F_N^{V_t} = \frac{1}{N^2} \log Z_N^{V_t}
		\end{equation*}

	\noindent
	admits the asymptotic expansion for $h\le K-1$

		\begin{equation}
			F_N(t) = \sum_{g=0}^h \frac{F_g(t)}{N^{2g}} + o\bigg{(} \frac{1}{N^{2h}} \bigg{)}.
		\end{equation}

	\noindent
	That is, for any $t \in (-\frac{7}{19008 2^{K-1}},\frac{7}{19008 2^{K-1}})$, we have

		\begin{equation*}
			\label{eqn:HCIZexpansion}
			F_N(t) = \sum_{g=0}^{K-1} \frac{F_g(t)}{N^{2g}} + r_N(t),
		\end{equation*}

	\noindent 
	where 

		\begin{equation*}
			\lim_{N \rightarrow \infty} \frac{r_N(t)}{N^{2(K-1)}} =0.
		\end{equation*}

	\noindent
	Moreover, there exists $\varepsilon > 0$ such that the coefficients
	$F_g(t)$ extend to analytic functions on the complex disc $|t|<\varepsilon$.

	Since $U(N)$ is compact, the partition function 
	$Z_N^{V_t}$ is an entire function of $t \in \field{C}$.  Thus $F_N(t)$, being the principal 
	branch of the logarithm of $Z_N(t)$, is analytic in a complex 
	neighbourhood of $t=0$.  Thus the error

		\begin{equation*}
			r_N(t) = F_N(t) -  \sum_{g=0}^h \frac{F_g(t)}{N^{2g}}
		\end{equation*}

	\noindent
	is a difference of analytic functions, and hence is
	also analytic in a neighbourhood of $t=0$.  Decomposing all relevant functions in 
	Maclaurin series,

		\begin{equation*}
			F_N(t) = \sum_{d=1}^\infty F_N^{(d)}(0)\frac{t^d}{d!}, \quad 
			F_g(t) = \sum_{d=1}^\infty F_g^{(d)}(0)\frac{t^d}{d!}, \quad
			r_N(t) = \sum_{d=1}^\infty r_N^{(d)}(0)\frac{t^d}{d!},
		\end{equation*}

	\noindent
	we thus have

		\begin{equation*}
			F_N^{(d)}(0) = \sum_{g=0}^h \frac{F_g^{(d)}(0)}{N^{2g}} + r_N^{(d)}(0)
		\end{equation*}

	\noindent
	for all $d \geq 1$, whence

		\begin{equation*}
			F_N^{(d)}(0) = \sum_{g=0}^h \frac{F_g^{(d)}(0)}{N^{2g}} + o\bigg{(} \frac{1}{N^{2h}} \bigg{)}
		\end{equation*}

	\noindent
	as $N \rightarrow \infty$, for each fixed $d \geq 1$.

	The asymptotics of the Maclaurin coefficients of $F_N(t)$ were obtained
	by a different method in \cite{GGN}, where a representation-theoretic argument
	was used to show that

		\begin{equation*}
			F_N^{(d)}(0) = \sum_{g=0}^h \frac{1}{N^{2g}} \sum_{\alpha,\beta \vdash d}
			(-1)^{\ell(\alpha)+\ell(\beta)} \sigma_g(x^\alpha) \sigma_g(y^\beta) \vec{H}_g(\alpha,\beta)
			+ o\bigg{(} \frac{1}{N^{2h}} \bigg{)},
		\end{equation*}

	\noindent
	with the $ \vec{H}_g(\alpha,\beta)$'s the monotone double 
	Hurwitz numbers.  Thus, from the uniqueness of asymptotic expansions
	on a given asymptotic scale (the scale here being $N^{-2}$), it follows that

		\begin{equation*}
			F_g^{(d)}(0) =  \sum_{\alpha,\beta \vdash d}
			(-1)^{\ell(\alpha)+\ell(\beta)} \sigma_g(x^\alpha) \sigma_g(y^\beta) \vec{H}_g(\alpha,\beta),
		\end{equation*}

	\noindent
	as required.

\bibliographystyle{amsplain}

\end{document}